\pdfoutput=1
\documentclass[%
onecolumn,
notitlepage,
reprint,
showpacs,
preprintnumbers,
showkeys,
amsmath,amssymb,
prx,
nofootinbib,
]{revtex4-2}

\makeatletter
\renewcommand{\p@subsection}{}
\renewcommand{\p@subsubsection}{}
\makeatother

\newcommand{\mypreamble}{
\title[\shortTITLE]{\TITLE}
\author{Ji\v{r}\'{\i} Proch\'{a}zka}
\email{jiri.prochazka@fzu.cz}
\affiliation{Institute of Physics of the Czech Academy of Sciences,\\ Na Slovance 1999/2, 18200 Prague 8, Czech Republic} 

}

\RequirePackage[T1]{fontenc}

\usepackage{amssymb}
\usepackage{amsmath}
\usepackage{latexsym}
\usepackage{amsthm}

\RequirePackage{graphicx}
\RequirePackage{mathptmx}      
\RequirePackage{flushend}

\usepackage{xifthen}

\usepackage{xr-hyper}
\ifthenelse{\isundefined{\flagEPJPLUS}}{%
\RequirePackage[
    colorlinks,
	citecolor=blue,
	urlcolor=blue,
	linkcolor=blue,
	bookmarksopen=true,
	bookmarksnumbered=true
]{hyperref} 
}{%
}

\usepackage{sidecap}

\ifthenelse{\isundefined{\flagEPJPLUS}}{%
\usepackage[toc,page]{appendix}
}{}
\usepackage[labelfont={small,bf},textfont=small]{caption}
\usepackage[labelfont={small,bf},textfont=small]{subcaption}

\usepackage{placeins} 

\usepackage{cases} 
\usepackage{empheq} 
\usepackage{enumitem}

\usepackage{expl3}
\usepackage{xparse}
\usepackage{l3keys2e}

\usepackage{siunitx}
\usepackage{gensymb}

\makeatletter
\def\blfootnote{\xdef\@thefnmark{}\@footnotetext}
\makeatother

\usepackage{keycommand}
\begingroup
  \makeatletter
  \catcode`\/=8 %
  \@firstofone
    {
      \endgroup
      \renewcommand{\ifcommandkey}[1]{%
        \csname @\expandafter \expandafter \expandafter
        \expandafter \expandafter \expandafter  \expandafter
        \kcmd@nbk \commandkey {#1}//{first}{second}//oftwo\endcsname
      }
   }

\ifthenelse{\isundefined{\flagEPJPLUS}}{%
\DeclareMathOperator{\eOp}{e}
\newcommand{\e}[1]{\eOp^{#1}}
}{}

\newcommand{\abs}[1]{\lvert {#1} \rvert}

\DeclareCaptionJustification{justified}{\justifying}
\usepackage{cleveref}  
\newcommand{\refp}[1]{(\ref{#1})}

\crefname{equation}{eq.}{eqs.}
\crefname{section}{sect.}{sects.}
\crefname{chapter}{chapter}{chapters}
\crefname{table}{table}{tables}
\crefname{figure}{fig.}{figs.}
\crefname{appsec}{appendix}{appendices}
\crefname{appchap}{appendix}{appendices}

\AtBeginEnvironment{appendices}{\crefalias{section}{appendix}}

\newtheorem{definition}{Definition}[section]
\newtheorem{assumption}{Assumption}[section]
\crefname{assumption}{assumption}{assumptions}

\newtheorem{theorem}{Theorem}[section]

\theoremstyle{remark}
\newtheorem{remark}{Remark}[section]

\newcommand{\TITLE}{Stochastic and corpuscular theory of (polarized) light}
\newcommand{\shortTITLE}{\TITLE}
\hypersetup{pdfauthor={Jiri Prochazka}, pdftitle={\shortTITLE}}

\graphicspath{{../arxiv_v3_20230119/figures/}}

\newcommand{\surface}{\ensuremath{\Sigma}}
\newkeycommand{\dElementSymbol}[event]{ \ensuremath{ \text{d}(\commandkey{event})}}

\newcommand{\indexset}{I}

\newcommand{\eventset}{\mathcal{F}}
\newcommand{\stochproc}{\{\Xw[i=i] : i \in \indexset \}}
\newcommand{\probabilityspace}{(\samplesetsymbol, \eventset,P)}
\newcommand{\varsymb}{X}
\newcommand{\wvarsymb}{\varsymb}
\newkeycommand{\Xw}[i=]{\wvarsymb_{\commandkey{i}}}
\newkeycommand{\dXw}[i=]{\text{d}\wvarsymb_{\commandkey{i}}}

\newkeycommand{\Xwin}[i=]{\wvarsymb_{\text{in}}^{\commandkey{i}}}
\newkeycommand{\dXwin}[i=]{\text{d}\wvarsymb_{\text{in}}^{\commandkey{i}}}

\newkeycommand{\Xwout}[i=]{\wvarsymb_{\text{out}}^{\commandkey{i}}}
\newkeycommand{\dXwout}[i=]{\text{d}\wvarsymb_{\text{out}}^{\commandkey{i}}}
\newkeycommand{\XwNR}[i=]{\wvarsymb^{NR}_{\commandkey{i}}}

\newkeycommand{\XwinNR}[i=]{\wvarsymb_{\text{in}\ifcommandkey{i}{,{\commandkey{i}}}{}}^{NR}}

\newkeycommand{\XwoutNR}[i=]{\wvarsymb_{\text{out}\ifcommandkey{i}{,{\commandkey{i}}}{}}^{NR}}
\newcommand{\ntrans}{M}  
\newcommand{\rangeXwNR}[2]{%
   \ifthenelse{ \equal{\detokenize{#1}}{\detokenize{#2}}}{\XwNR[i=#1]}{
		   \ifthenelse{ \equal{\detokenize{#1}}{\detokenize{0}} \AND \equal{\detokenize{#2}}{\detokenize{1}} }{\XwNR[i=#1], \XwNR[i=#2]}{
				   \ifthenelse{ \equal{\detokenize{#1}}{\detokenize{#2}} }{\XwNR[i=#1]}{\XwNR[i=#1], \dotsc, \XwNR[i=#2]}
		   }
   }
}

\newkeycommand{\ifcomma}[i=]{\ifcommandkey{i}{,{\commandkey{i}}}{}}

\newcommand{\samplesetsymbol}{\Omega}

\newcommand{\statespacesymbol}{S}
\newcommand{\stateelementsymbol}{s}

\newcommand{\statesetmeasure}[1]{\mathcal{S}_{#1}}
\newcommand{\stateset}[1]{\statespacesymbol_{#1}}

\newcommand{\statee}[1]{\stateelementsymbol_{#1}}

\DeclareMathOperator{\dos}{dos}
\newcommand{\dosstateSi}[1]{\dos_{\statee{#1}}}
\newcommand{\dosargstateSi}[1]{\dosstateSi{#1}(\Xw[i=#1])}

\newcommand{\dosargSi}[1]{\dos_{#1}(\Xw[i=#1])}

\ExplSyntaxOn
\keys_define:nn { theoryouti }
 {
  i .tl_gset:N = \g_theoryouti_i_tl,
 }

\ProcessKeysOptions{ theoryouti }

\NewDocumentCommand{\stateSoutCondASi}{ O{} } 
{
  \keys_set:nn { theoryouti }
  {
    #1 
  }
  \statee{\g_theoryouti_i_tl}\mid \stateset{\g_theoryouti_i_tl}
}
\ExplSyntaxOff

\newkeycommand{\rhoS}[i=]{\rho_{\stateset, \commandkey{i}}}

\newcommand{\rhoargSi}[1]{\rhoS[i=#1](\Xw[i={#1}])} 
\newcommand{\rhoSi}[1]{\rhoS[i=#1]} 

\newcommand{\Ni}[1]{N_{#1}}

\newcommand{\ProbTi}[1]{P_{T\ifcomma[i=#1]}}
\newcommand{\ProbTiarg}[1]{\ProbTi{#1}(\Xw[i={#1}])} 
 
\newcommand{\rhoCindexsymb}{C}
\newcommand{\rhoCi}[1]{\rho_{\rhoCindexsymb\ifcomma[i=#1]}}
\newcommand{\rhoCiarg}[1]{\rhoCi{#1}(\Xw[i={#1}], \Xw[i={#1+1}])} 

\newkeycommand{\rhoCindep}[i=]{\rho_{\widetilde{\rhoCindexsymb}\ifcommandkey{i}{,{\commandkey{i}}}{}}}   

\newcommand{\letbeprocess}[1]{Let $\stochproc$ be stochastic process given by \cref{#1}.}

\newcommand{\refsasprocess}{}
\newcommand{\refsasprocessfull}{}

\newcommand{\refsasprocessSTCmain}{as:independent_realizations,as:markov_property}
\newcommand{\refsasprocessSTCmainfull}{process-stc-as:independent_realizations,process-stc-as:markov_property}
\newcommand{\refsasprocessSTCextra}{as:different_state_spaces,as:state_prob_transition,\refsasprocessSTCmain}
\newcommand{\refsasprocessSTCextrafull}{process-stc-as:different_state_spaces,process-stc-as:state_prob_transition,\refsasprocessSTCmainfull}
\newcommand{\refsasprocessSTC}{\refsasprocess,\refsasprocessSTCextra}
\newcommand{\refsasprocessSTCfull}{\refsasprocessfull,\refsasprocessSTCextrafull}

\newcommand{\refsasprocessSTCsequenceextra}{}
\newcommand{\refsasprocessSTCsequenceextrafull}{}
\newcommand{\refsasprocessSTCsequence}{\refsasprocessSTC,\refsasprocessSTCsequenceextra}
\newcommand{\refsasprocessSTCsequencefull}{\refsasprocessSTCfull,\refsasprocessSTCsequenceextrafull}

\newcommand{\refsasprocessSTCsequencecaseoneexttrafull}{process-stc-as:rhoC_independent_of_xiwin}

\newcommand{\refsasprocessSTCsequencecaseonefull}{\refsasprocessSTCsequencefull,\refsasprocessSTCsequencecaseoneexttrafull}

\newcommand{\refsasprocessSTCsequencecasetwoextrafull}{process-stc-as:rhoC_identity}

\newcommand{\refsasprocessSTCsequencecasetwofull}{\refsasprocessSTCsequencecaseonefull,\refsasprocessSTCsequencecasetwoextrafull}

\DeclareMathOperator\erf{erf}
\DeclareMathOperator\sym{sym}

\newkeycommand{\polarization}[i=]{ \ensuremath{\theta_{\commandkey{i}}}}
\newkeycommand{\rotangle}[i=]{ \ensuremath{\alpha_{\commandkey{i}} }}

\newcommand{\rhoSpol}[1]{\rho_{S\ifcomma[i=#1]}}
\newcommand{\rhoSargipol}[1]{\rhoSpol{#1}(\polarization[i=#1],\vec{\rotangle}) }

\newcommand{\ProbTpol}{\ProbTi{}}

\newcommand{\ProbTidentityargpol}{\ProbTpol(\polarization[i=\text{in}], \rotangle)}
\newcommand{\ProbTargipol}{\ProbTpol(\polarization[i=i], \rotangle[i=i])}
\newcommand{\ProbTiargipol}{\ProbTi{i}(\polarization[i=i], \rotangle[i=i])}
\newcommand{\ProbTargiexpandedpol}{\ProbTpol(\polarization[i=\text{in}]\!=\!\polarization[i=i], \rotangle\!=\!\rotangle[i=i])}

\newcommand{\rhoCpol}{\rhoCi{}}
\newcommand{\rhoCargpol}{\rhoCpol(\polarization[i=\text{in}], \polarization[i=\text{out}], \rotangle)}
\newcommand{\rhoCargipol}{\rhoCpol(\polarization[i=i], \polarization[i={i+1}], \rotangle[i=i])}
\newcommand{\rhoCiargipol}{\rhoCi{i}(\polarization[i=i], \polarization[i={i+1}], \rotangle[i=i])}
\newcommand{\rhoCargexpandedpol}{\rhoCpol(\polarization[i=\text{in}]\!=\!\polarization[i=i], \polarization[i=\text{out}]\!=\!\polarization[i={i+1}], \rotangle\!=\!\rotangle[i=i])}

\newcommand{\rhoCindepargipol}{\rhoCindep[i=i](\polarization[i={i+1}], \rotangle[i=i])}
\newcommand{\rhoCindepidentityargpol}{\rhoCindep(\polarization[i=\text{out}], \rotangle)}

\newcommand{\rhoCindepidentityargiexpandedpol}{\rhoCindep(\polarization[i=\text{out}]\!=\!\polarization[i={i+1}], \rotangle\!=\!\rotangle[i=i])}

\newcommand{\dosargipol}[1]{\dos_{#1}(\polarization[i=#1],\vec{\rotangle})}

\newcommand{\parameterizedFunctionspol}{$\rhoSargipol{0}$, $\ProbTidentityargpol$ and $\rhoCargpol$}   

\newlist{openquestions}{enumerate}{1}
\setlist[openquestions,1]{
    label={\arabic*.},
    ref=\arabic*, 
	} 
\creflabelformat{openquestionsi}{#2\textup{#1}#3} 
\crefname{openquestionsi}{open question}{open questions} 
\Crefname{openquestionsi}{Open question}{Open questions}

\newlist{proposedsolutions}{enumerate}{3}
\setlist[proposedsolutions,1]{
    label={\arabic*.},
    ref=\arabic*, 
	} 
\setlist[proposedsolutions,2]{label=\arabic*., ref=\arabic{proposedsolutionsi}.\arabic*} 
\setlist[proposedsolutions,3]{label=\arabic*., ref=\arabic{proposedsolutionsii}.\arabic*} 
\creflabelformat{proposedsolutionsi}{#2\textup{#1}#3} 
\creflabelformat{proposedsolutionsii}{#2\textup{#1}#3} 
\creflabelformat{proposedsolutionsiii}{#2\textup{#1}#3} 
\crefname{proposedsolutionsi}{proposed solution}{proposed solutions} 
\Crefname{proposedsolutionsi}{Proposed solution}{Proposed solutions}
\crefname{proposedsolutionsii}{proposed solution}{proposed solutions} 
\Crefname{proposedsolutionsii}{Proposed solution}{Proposed solutions}
\crefname{proposedsolutionsiii}{proposed solution}{proposed solutions} 
\Crefname{proposedsolutionsiii}{Proposed solution}{Proposed solutions}

\begin{document}
\mypreamble

\newcommand{\myabstract}{
Both the corpuscular theory of light and the theory of stochastic processes are well known in literature. However, they are not systematically used together for description of optical phenomena. There are optical phenomena, such as the well known three-polarizers experiment or other phenomena related to polarization of light, which have never been quantitatively and qualitatively explained using the concept of quantum of light (photon). The situation changed in 2022 when stochastic memoryless and independent (IM) process formulated within the framework of the theory of stochastic processes was introduced. It is suitable for determination of probability (density) functions characterizing interaction (transmission or reflection) of individual photons with optical elements on the basis of experimental data. The process has memoryless (Markov) property and it is assumed that the interactions of individual photons with an optical system are independent. Formulae needed for analysis of data in the context of polarization of light are derived. An example analysis of the three-polarizers experiment is performed and numerical result of the probability (density) functions are determined. These original results were missing in literature. The results imply that the possibilities of the corpuscular theory of light to describe optical phenomena can be significantly extended with the help of stochastic IM process and the theory of stochastic processes in general. 
}

\ifthenelse{\isundefined{\flagEPJPLUS}}{%
\begin{abstract}
\myabstract
\end{abstract}
}{%
\abstract{\myabstract}

}

\keywords{
the corpuscular theory of light, 
the theory of stochastic processes, 
stochastic independent and memoryless (Markov) process, 
stochastic IM process, 
optics, 
photon, 
optical systems,
polarization of light, 
three-polarizers experiment 
}

\maketitle

\clearpage
\ifthenelse{\isundefined{\flagEPJPLUS}}{%
\tableofcontents
}{%
}

\section{\label{sec:introduction}Introduction}

The question of whether light is corpuscular or wave in nature is very old and has been debated in the literature for centuries. Sometimes the corpuscular theory and the wave theory are understood as "either-or" and sometimes as "both-and" (wave-particle dualism). We will not attempt to answer this question once and for all. We will only focus on the possibilities of the corpuscular theory of light itself.

I.~Newton tried to explain a number of phenomena with light by using the idea that light consists of particles (corpuscles). However, determination of properties of these particles varied greatly with time; useful historical references and many comments to different models of light quanta (photons) are in \cite{Hentschel2018}. The interaction of light with matter is still the subject of extensive scientific research. One of many important discoveries related to particle of light concerned its energy. M.~Planck in 1900 \cite{Planck1900,Planck1901} introduced quantization of energy of electromagnetic radiation to describe spectrum of black-body radiation. A.~Einstein in 1905 \cite{Einstein1905quanta} developed further this quantum hypothesis. He basically introduced modern concept of particle of light - photon - and used it to explain the photoelectric effect. 
Photon has been successfully used for explanation of many other observed phenomena related to light since the beginning of the 20th century. It has led to numerous technological advances which further demonstrate in various ways the usefulness of the concept of particles of light (photons). There is now whole field called photonics \cite{Photonics2015vol1,Photonics2015vol2,Photonics2015vol3,Photonics2015vol4} which began with the invention of the laser in 1960 and which further demonstrate the practical applicability of photons.

However, there are phenomena in optics which have never been described both quantitatively and qualitatively using the corpuscular idea of light. This concerns, e.g., experiments with polarization of light. Interesting properties of light related to polarization may be easily demonstrated with the help of the well known and discussed three-polarizers experiment. When incident unpolarized light passes through two crossed linear polarizers then the intensity of light is strongly reduced (nearly no light passes through). However, when another polarizer is placed between the two polarizers then some light can pass through the sequence of the three polarizing filters. The intensity of light behind the third polarizer in the sequence is maximal when the second polarizer has polarizing axis oriented at $45 \si{\degree}$ relative to the polarizing axis of the first polarizer (and the third polarizer). This experimental result, which can be easily reproducible and observed by naked eye, has attracted attention of many researches as it is clear that the effect cannot be explained by only absorption of some photons from the beam by the polarizers; the polarizers must change also a property of the transmitted photons (called \emph{polarization}). Moreover, both the effects must depend on the orientation of the axes of the polarizers. Similar effects can be observed when light is transmitted through other polarization sensitive elements.

The situation changed significantly in 2022 when stochastic independent and memoryless (Markov) process (IM process) formulated within the theory of stochastic processes was introduced in \cite{Prochazka2022statphys_process_stc}. The independence is related to assumed independence of outcomes (realization) of an experiment. It has opened completely new possibilities to describe a whole range of (not only) particle phenomena that have not yet been described particle-wise. IM process can describe under the two assumptions: particle decays, motion of particles when their initial conditions are specified statistically, particle-matter interactions and particle-particle collisions in unified way \cite{Prochazka2022statphys_process_stc}.

Several statistical descriptions of optical phenomena involving photons have been, implicitly or explicitly, based on the two assumptions. Description based on IM process can be, therefore, reduced to these relatively simple descriptions. One may ask how IM process (the theory of stochastic processes in general) can help to describe optical phenomena (such as the three-polarizer experiment) which have not yet been described by using the corpuscular theory of light and statistics. The theory of stochastic processes has very wide applications. Many stochastic processes having various properties are studied in literature. However, optics is one of the field where it has been used only partially or not at all. We will try to show that this theory can be very useful in explaining and describing phenomena in optics, as well as it is very helpful in other areas of scientific research.

This paper is structured as follows. The statement of the problem is discusses in greater detail in \cref{sec:problem_statement}. \Cref{sec:measurement} briefly summarizes measurement of number of photons (beam intensities) transmitted through a sequence of $\ntrans$ optical elements in dependence on their rotations and the results concerning linear polarizers. Contemporary statistical theoretical approaches used for description of polarization of light and light in general are discussed in \cref{sec:contemprorary_theories}. They contain several limitations that in many cases make them unsuitable for analyzis of experimental data and achieving the stated goals. These limitations can be removed using stochastic IM process (or other suitable process formulated within the theory of stochastic processes) which is explained in \cref{sec:theory_of_stochastic_processes}. IM process describing transmission of individual photons through $\ntrans$ polarization sensitive elements (including, but not limited to, linear polarizers) is formulated in \cref{sec:model}. Several ideas of M.~V.~Lokaj\'{i}\v{c}ek \cite{lokajicek2002quantum} will be developed further. Probability (density) functions characterizing transmission of individual photons through optical elements can be used for definition of properties of the beam and various types of optical elements. Definitions related to polarized light and various types of polarization sensitive elements are discussed in \cref{sec:definitions}. An example data corresponding to transmission of light through 3 linear polarizers are analyzed with the help of IM process in \cref{sec:data_analysis} where numerical results are shown, too. The example analysis of data in \cref{sec:data_analysis} leads to probability (density) functions which are not determined uniquely. What to do in similar situations is discussed in \cref{sec:open_questions_answers}. Summary and concluding remarks are in \cref{sec:conclusion}.





\section{\label{sec:problem_statement}Problem statement}

Let us consider a sequence of $\ntrans$ optical elements (they may or may not be polarization sensitive), as shown in \cref{fig:optical_elements}. This is arrangement of many  experiments (in general a net of optical elements can be consider). The interaction of a photon through $i$-th element ($i=0,...,\ntrans-1$) can be considered stochastic (random) process (a photon may or may not pass through given element, it can be reflected, ...). We may ask how to determine probability (density) functions characterizing interaction of photons with the individual elements in the sequence. We will commonly use the term transmission to refer to any interaction of a photon (beam) with an optical element.

According to \cite{Prochazka2022statphys_process_stc} (see eq.~(63) therein) probability of transition of a system from given state to another given state can be factorized into 3 conditional probabilities (assuming Markov property and independence of outcomes of an experiment). In the context of optics the probability of transmission of given input photon state through $i$-th sensitive optical element to given output photon state can be factorized into 3 conditional probabilities. I.e., 3 probabilistic effects can be distinguished:
\begin{enumerate}
\item{an incoming photon before interacting with $i$-th optical element has its state specified by a probability density function $\rhoSi{i}$;}
\item{an incoming photon in given state may or may not be transmitted through $i$-th element (characterized by a probability function $\ProbTi{i}$ being function of the incoming photon state, not the output photon state);}
\item{if the photon in given state is transmitted through the element then its state may or may not change (described by a probability density function $\rhoCi{i}$ being function of both the incoming and the outgoing photon state).}
\end{enumerate}
Both incoming and outgoing photon states are in general described by random and non-random variables. Optical elements can be characterized by random and non-random variables, too. 

One may ask how the probability (density) functions $\rhoSi{i}$, $\ProbTi{i}$ and $\rhoCi{i}$ characterizing transmission of individual photons with  given $i$-th optical element can be determined on the basis of experimental data. It seems that it has never been done systematically.


\begin{figure*}
\captionsetup{width=0.95\textwidth}
\centering
\includegraphics*[width=\textwidth]{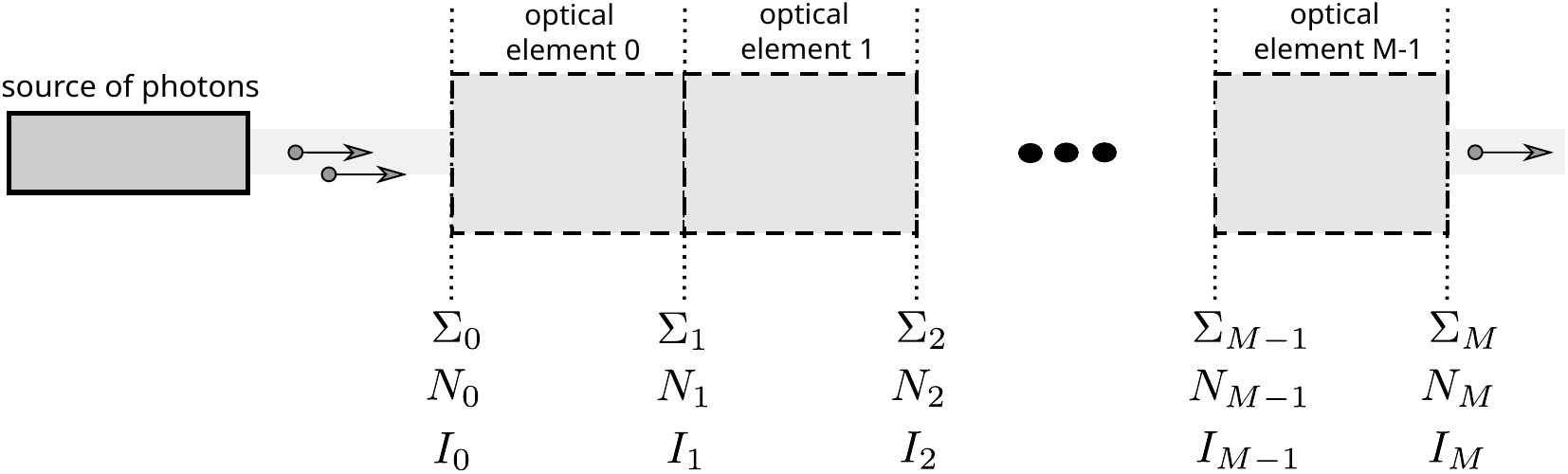} 
\caption{\label{fig:optical_elements}Scheme of sequential arrangement of many optical systems (experiments with light). A beam of photons is passing through control surfaces $\surface_{i}$ ($i \in (0,1, \dotsc, \ntrans)$) which define $\ntrans$ transport segments, each of them containing an optical element.
}
\end{figure*}

\section{\label{sec:measurement}Measurement}
It is an experimental goal to measure the beam characteristics before and after interaction with a given optical element ("sample"). Important experimental data about the sample can be obtained when input (diagnostic) beams of various properties and the properties of the output beams are measured, see sect.~2 in \cite{Prochazka2022statphys_process_stc}. I.e., response of an element to various inputs can be measured. In general transmission of a photon through an optical element depends on random and non-random variables. State of input and output beam can be characterized by quantity called density of photon states which depends on the random and non-random variable. Density of states will be defined later in \cref{sec:theory_of_stochastic_processes_generalities}. In many cases it is not possible to measure it as a function of all the random and non-random variables.

Let us now assume that only the number of photons before and after an interaction with each optical element in a sequence can be measured as a function of non-random variables. This is discussed in \cref{sec:measurement_general} in the case the non-random variables correspond to the rotations of the elements (a common case when measuring polarization-sensitive elements). In \cref{sec:measurement_polarizers} one can find the dependences in the case of sequence of linear polarizers.

\subsection{\label{sec:measurement_general}Measurement of number of photons passing through sequence of optical elements in dependence on their rotations}
\subsubsection{Experimental setup}
Let us consider photon beam passing trough a sequence of $\ntrans$ optical elements, see \cref{fig:optical_elements}. 
Control surfaces $\surface_i$ ($i \in (0,1, \dotsc, \ntrans)$) can be introduced. Each optical element is placed between two neighbouring control surfaces. These control surfaces do not correspond to physical surfaces of the optical elements, but to places where characteristics of the beam can be measured (they will be called surfaces for short).

\subsubsection{\label{sec:measurement_particle_number}Number of photons}
Let $N_{i}$ be the \emph{number of photons} which passed through $i$-th surface $\surface_i$. The number of photons passed through $i$-th optical element, i.e., $N_{i+1}$ ($i=0,...,\ntrans-1$ in this case), may in general depend on various random and non-random variables characterizing states of the optical elements and states of photons. Let us assume, for the sake of simplicity, that it depends only on rotation of the $i$-th optical element, the rotations of all the other optical elements placed in front of it, and the number of initial particles $N_{0}$. We may introduce vector $\vec{\rotangle}$ having components $\rotangle[i=i]$ ($i=0,...,\ntrans-1$) and representing rotations of all the individual optical elements. For the sake of simplicity we introduce rotation of an optical elements characterized by only one parameter (non-random variable), but two other parameters fully specifying rotation of given element in space could be introduced, too. It is useful to introduce a convention that if a function depends on $\vec{\rotangle}$ then it may depend only on some of its components.

\emph{Transmittance} $T_{i}(\vec{\rotangle})$ of $i$-th optical element ($i=0, ..., \ntrans-1$) is standardly defined as  
\begin{subnumcases}{ T_i(\vec{\rotangle}) = \label{eq:T_i_def} }
		\frac{N_{i+1}(\vec{\rotangle})}{N_{i}(\vec{\rotangle})}  & \text{if $N_{i}(\vec{\rotangle})\neq 0$} \\
                                           0 & \text{otherwise} 
\end{subnumcases}
It characterizes how much $i$-th optical element transmits light in dependence on its rotation and rotations of all the preceding optical elements. 

\begin{remark}
Number of photons transmitted through an optical element may in general depend, e.g., on time. In some cases this experimental information must be taken into account to correctly interpret the experimental results. However, in many cases (e.g., in the case of linear polarizers which will be discussed below) the measured quantities often do not depend on time (numbers of transmitted photons corresponding to the initial number of photons $N_0$ are counted behind each optical element in a sequence independently of the time when the photons were transmitted). Time variable will not be of our interest in the following.
\end{remark}

\subsubsection{\label{sec:measurement_intensity_number}Number of photons and beam intensity}
\begin{assumption}[Energy of photons]
\label{as:energy_pol}
Photons have the same energy when they pass through surface $\surface_{i}$ for all $i=0,\dotsc,\ntrans$.
\end{assumption}

Sometimes \emph{intensity} of a particle beam defined as energy incident on a surface per unit of time and per unit of area (i.e., having units of $Js^{-1}m^{-2}=Wm^{-2}$)\footnote{In radiometry this definition of intensity corresponds to irradiance (see chapter 34 in \cite{HandbookOfOptics2009_vol2}).} is measured in an experiment. The intensity of light beam passing through surface $\surface_i$ may be denoted as $I_{i}(\vec{\rotangle})$ ($i=0,...,M$), see \cref{fig:optical_elements}. The intensity $I_0$ is the initial intensity of the beam. If the intensity $I_{i}$ and the number of photons $N_{i}$ correspond to the same surface area and time interval \emph{and the photons have the same energy} (see \cref{as:energy_pol})  then it holds (if $N_{i}(\vec{\rotangle})) \neq 0$, $i=0,...,\ntrans-1$)
\begin{align}
		\frac{N_{i+1}(\vec{\rotangle})}{N_{i}(\vec{\rotangle})} 
		&=  \frac{I_{i+1}(\vec{\rotangle})}{I_{i}(\vec{\rotangle})}  \label{eq:relation_N_to_I_ip1_to_i}
\end{align}
and (if $N_{0} \neq 0$, $i=0,...,M$)
\begin{align}
		\frac{N_{i}(\vec{\rotangle})}{N_{0}} &= \frac{I_{i}(\vec{\rotangle})}{I_{0}}   
		\, . \label{eq:relation_N_to_I_i_to_0}
\end{align}

\begin{remark}
\Cref{as:energy_pol} could be, for the purposes presented in this paper and for the purposes to relate the relative numbers of transmitted photons to the relative beam intensities, replaced by assumption that \cref{eq:relation_N_to_I_i_to_0} holds.
\end{remark}

\subsection{\label{sec:measurement_polarizers}Example - sequence of linear polarizers}
\begin{figure*}
\captionsetup{width=0.95\textwidth}
\centering
\includegraphics*[width=\textwidth]{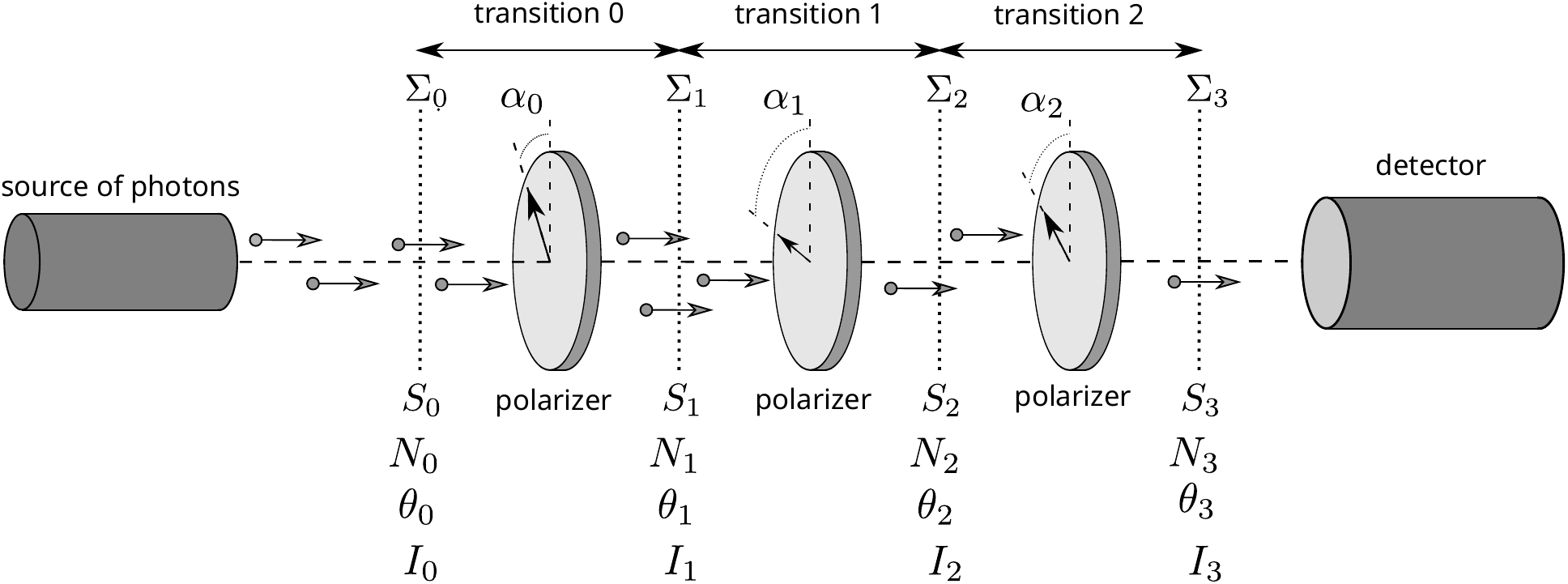}
\caption{\label{fig:experimental_setup_polarizers} Example of experimental setup for measurement of relative photon numbers $N_{i}/N_0$ (beam intensities $I_{i}(\vec{\rotangle})/I_{0}$) for $i=1,2,3$ in dependence on the rotation angles $\vec{\rotangle}=(\rotangle[i=0],\rotangle[i=1],\rotangle[i=2])$ of the linear polarizers when light beam is transmitted through one, two or three polarizers ($M=3$). The measured relative photon numbers correspond to conveniently chosen control surfaces $\surface_i$ (they may not correspond to physical surfaces of the polarizers). Values of polarization $\polarization[i=i]$ of photons which passed through surface $\surface_i$ form state space $\stateset{i}$. The values of photon polarization angles must be determined with the help of a probabilistic model on the basis of experimental data.
}
\end{figure*}

Let us now consider a photon beam passing through $\ntrans$ linear polarizers, see \cref{fig:experimental_setup_polarizers} corresponding to $\ntrans=3$, as an example of transmission of light through a sequence of polarization sensitive elements.

\subsubsection{\label{sec:ideal_polarizers}Ideal linear polarizers and Malus's law}
In textbooks on optics transmission of light through ``ideal'' linear polarizers is often discussed (see, e.g., sect.~12.4 in \cite{HandbookOfOptics2009_vol1}). If initially unpolarized light beam is sent through a sequence of $\ntrans$ ideal linear polarizers then the transmittance $T_{i}(\vec{\rotangle})$ of $i$-th ideal polarizer ($i=0,...,\ntrans-1$) is
\begin{subnumcases}{T_{i}(\vec{\rotangle}) = \label{eq:T_i_ideal_polarizer}}
                          \frac{1}{2} & \text{if $i = 0$} \label{eq:T_i0_ideal_polarizer} \\
		\cos^2(\rotangle[i=i] - \rotangle[i=i-1]) & \text{if $0 < i < M$}    \, .
\end{subnumcases}
I.e., the initial number of photons $N_0$ after transmission through the first ideal polarizer drops to one half independently of the orientation of the polarization axis $\rotangle[i=0]$. The cosine-squared function expresses the known Malus's law. It represents the first quantitative relationship for treating polarized light intensities describing measurements performed by É.-L.~Malus \cite{Malus1809malus_law} (see also summary of the life of Malus, the historical context, and further comments concerning his law in \cite{Kahr2008}). The transmittances $T_{i}$ given by \cref{eq:T_i_ideal_polarizer} do not depend on the value of the initial number of particles $N_{0}$.

\Cref{eq:T_i_def,eq:T_i_ideal_polarizer} imply that numbers of transmitted photons corresponding to the surface $\surface_i$ divided by initial number of photons is
\begin{subequations}
		\label{eq:N_i_over_N_0_ideal_polarizer}
\begin{empheq}[left={\frac{N_{i}(\vec{\rotangle})}{N_0}=\empheqlbrace\,}]{align}
		&             1  & \text{if $i = 0$}  \\
		&                                             \frac{1}{2} & \text{if $i = 1$}  \label{eq:N_i_over_N_0_i1_ideal_polarizer}  \\
		&\frac{1}{2}	\Pi^{i-1}_{j=1}\cos^2(\rotangle[i=j] - \rotangle[i=j-1]) & \text{if $1 < i < M$}  \label{eq:N_i_over_N_0_igt1_ideal_polarizer}  \, . \!\!\!\!
\end{empheq}
\end{subequations}

\subsubsection{\label{sec:real_polarizers}Real linear polarizers}
Number of photons transmitted through a sequence consisting of  ``real'' linear polarizers may be very similar to the number of photons transmitted through sequence of ideal polarizers, see \cref{sec:ideal_polarizers}. However, it is known that some other real polarizers differ significantly from the ideal polarizers. Every two real linear polarizers transmit some light, even if they are crossed, and absorb some light if their axes are parallel. This is often expressed by modifying the Malus's law, as discussed in sect.~12.4 in \cite{HandbookOfOptics2009_vol1}.  
In \cite{Krasa1993} it has been measured that two crossed real linear polarizers produce unpolarized light (value of $N_{3}(\vec{\rotangle})/N_0$ behind the third, testing, linear polarizer is independent of its rotation angle). It has been measured in \cite{Krasa1994} that transmittance of a pair of real linear polarizers may have local maximum if the polarizer axes are mutually crossed, and around this point the transmittance is not fully symmetrical (the effect can be measured with sensitive measurement devices). Some smaller or bigger deviations of measured beam intensity from Malus's law are visible in experiments conducted in \cite{Kadri2014,Vertchenko2016,Monteiro2017}, too. The deviations depend on properties of the polarizers. Some other differences of real linear polarizers from ideal ones are summarized in sect.~15.27 in \cite{HandbookOfOptics2009_vol1}. It has not been possible to observe many of these differences by Malus at the beginning of 19th century, i.e., before the advent of photomultiplier tubes (PMTs), charge-coupled devices (CCDs) and other devices nowadays commonly used for measurement of numbers of photons (beam intensities).

\section{\label{sec:contemprorary_theories}Contemporary widely used descriptions of polarization of light}
The study of interaction of light with matter has a very long tradition. Several distinct theoretical approaches have been developed to better understand huge amount of observed phenomena \cite{Goldstein2011,Goodman2015statistical_optics,Rice2020quantum_optics,Weiner2003light_matter,Hartmann2005theoretical_optics,Menn2004practical_optics,Thorne2017modern_classical_physics,Meschede2017optics_light_lasers,HandbookOfOptics2009_vol1,Photonics2015vol1,BornWolf2013principles_of_optics}. 

The measurement and description of effects related only to polarization of light represent a very broad topic, see \cite{Shurcliff1962polarized_light,Clarke1971polarized_light,Azzam1977,Brosseau1998fundamentals_of_polarized_light,Tompkins2005ellipsometry,Fujiwara2007spectroscopic_ellipsometry} and Part 3 in \cite{HandbookOfOptics2009_vol1} devoted to optical polarization. A wide range of polarization phenomena, their applications and the history of the polarization of light are discussed in \cite{Horvath2003}.

Theoretical descriptions of polarization of light can be divided into two groups, depending on whether or not they use statistics. Theoretical approaches which use statistics only partially or not at all are discussed in \cref{sec:contemprorary_theories_non_statistical} and we will not pay much attention to them. \Cref{sec:contemprorary_theories_statistical} deals with theoretical methods, which are much more related to statistics. Summary of the theoretical approaches, their possibilities to determine probability density functions $\rhoSi{i}$, $\ProbTi{i}$ and $\rhoCi{i}$; are in \cref{sec:contemprorary_theories_summary}.

\subsection{\label{sec:contemprorary_theories_non_statistical}Non-statistical descriptions}
Various matrix calculi have been developed to describe the effect of polarization sensitive elements on the state of polarization of a light beam, see sect.~12.8~in~\cite{HandbookOfOptics2009_vol1} and, e.g., a historical revision of the development of the differential Stokes–Mueller matrix formalism \cite{Arteaga2017}. The most widely known matrix calculi are the Mueller calculus and the Jones calculus. The distinct matrix formalisms have been developed with different aims and purposes in order to better understand observed polarization phenomena. However, none of the matrix calculi has been designed to describe probabilistic character of transmission of individual photons through a sequence of polarization sensitive elements, such as the mentioned sequence of 3 linear polarizers. They describe some average transmission of the whole beam (a state of the beam being represented by a vector). This is common feature of many matrix descriptions used in optics \cite{Gerrard1975}. E.g., the description of transmission of light through sequence of three polarizers by Jones \cite{Jones1956} does not mention probability at all. Similarly, description of three linear polarizers in the context of astrophysics discussed in \cite{DeForest2022} is also not focused on probabilistic description of transmission of individual photons. This is also the case of the Maxwell's equations. None of these non-statistical approaches of description of polarization of light is suitable for determination of the probability (density) functions $\rhoSi{i}$, $\ProbTi{i}$ and $\rhoCi{i}$.

\subsection{\label{sec:contemprorary_theories_statistical}Statistical descriptions}
Several diverse theoretical approaches more strongly associated to statistics (the probability theory) have been developed in optics, too. \emph{Statistical optics} \cite{Goodman2015statistical_optics} deals mainly with a scalar theory of light waves. The scalar quantities are regarded as representing one polarization component of the electric or magnetic field, with the assumption that all such components can be treated independently. Its application to data concerning propagation of light phenomena is based on several other assumptions (such as assumptions accompanied by introduction of complex amplitudes and their properties). 

This theoretical approach does not fully explore the concept of quantum of light in which we are interested in. This is also the case of other contemporary theoretical attempts \cite{Brosseau1998fundamentals_of_polarized_light} trying to describe various polarization (optical) phenomena with the help of statistics. \emph{Quantum mechanics} in general can take into account individual photons. However, description of polarization of light in \emph{quantum optics} discussed, e.g., in sect.~6 in \cite{Agarwal2013quantum} does not provide straightforward way to determine the functions $\rhoSi{i}$, $\ProbTi{i}$ and $\rhoCi{i}$ characterizing transmission of individual photons through polarization sensitive elements (or other optical elements). Description of polarization of light based on quantum optics and discussed in \cite{Schranz2021} neglects polarization-dependent losses during transmission of photons through a medium and is based on several other limiting assumptions. It is, therefore, not suitable for description of experiment with 3 linear polarizers where it is necessary to take into account that a photon can be absorbed by one of the linear polarizers. There is also \emph{quantum state tomography (QST)} which refers to any method that allows one to reconstruct the accurate representation of a quantum system based on data obtainable from an experiment. This is the kind of method we are looking for. However, QST introduces complex probability amplitudes (wave-functions), as in quantum mechanics, and also in this case it seems that the probability (density) functions $\rhoSi{i}$, $\ProbTi{i}$ and $\rhoCi{i}$ have never been determined in the case of three-polarizers experiment (or similar one) \cite{Toninelli2019,Czerwinski2022}.  

\emph{Monte Carlo (MC) methods} represent another broad group of theoretical statistical methods used for description of particle transport phenomena. They are summarized in \cite{Haghighat2015}. They use random or pseudorandom numbers to simulate a random process. To make a MC simulation of transmission of particles through a medium it is necessary to know parameters of various microscopic processes. Determination of the probability density functions characterizing the processes typically requires detailed knowledge of cross sections of various interaction types, mean free paths, spatial distribution and types of scattering centers, etc. The parameters of the microscopic processes and the dependence of the processes on various random variables are often now known. This kind of information need to be determine first on the basis of experimental data. It makes MC methods hardly usable for determination of the probability (density) functions $\rhoSi{i}$, $\ProbTi{i}$ and $\rhoCi{i}$ characterizing \emph{overall} effect of transmission of individual photons through $i$-th optical element.   

\subsection{\label{sec:contemprorary_theories_summary}Summary}
One can look similarly over other theoretical methods used in optics. It is possible to make several observations concerning the contemporary statistical descriptions widely used in optics:
\begin{enumerate}
\item{States of individual photons are not always taken into account (some methods are more focused on average properties).}
\item{Evolution (transition) operator is standardly assumed to be unitary in quantum mechanics. This makes it delicate to describe phenomena corresponding to probability $\ProbTi{i}$ which is not identically equal to 1 and depends on random variables characterizing given initial state (for further details see the open problem 4 in sect.~6 in \cite{Prochazka2016} concerning contemporary descriptions of particle collisions).}
\item{Some theoretical methods require information about a system (such as the Hamiltonian or parameters of various microscopic processes) which may not be known and need to be determined first on the basis of experimental data.} 
\item{Introduction of complex probability amplitudes or wave-functions is often accompanied by introduction of additional assumptions concerning their properties. We are, roughly speaking, interested mainly in the probabilities given by the square of the absolute value of the amplitudes. Introduction of the amplitudes may, therefore, bring complications rather than benefits.}
\item{It is often difficult to follow under which assumptions various statements are made.}
\item{Only one transition of a photon state to another photon state is often discussed. Sequences (or even nets) of transitions are rarely mentioned.}
\item{The probability of transition of a given initial state to a given final state is not factorized into the 3 probabilistic effects mentioned in \cref{sec:problem_statement}.}
\end{enumerate}
We may conclude that the probability (density) functions $\rhoSi{i}$, $\ProbTi{i}$ and $\rhoCi{i}$ have never be determined on the basis of experimental data. The contemporary theoretical approaches used in optics do not provide straightforward way to do it. 

\section{\label{sec:theory_of_stochastic_processes}Theory of probability and theory of stochastic processes}
A branch of mathematics concerned with the analysis of random phenomena is the probability theory. Stochastic or random process is a mathematical object usually defined as a family of random variables. The theory of stochastic processes provides very general and abstract framework to study random processes, experiments having random outputs. It may be seen as extension of the probability theory. Properties of many stochastic processes are intensively studied in literature and they have numerous applications in basically all field of research. Optics is, however, one of the field where it is used only partially or not at all (see \cref{sec:contemprorary_theories}).

\Cref{sec:theory_of_stochastic_processes_generalities} is devoted to definitions concerning general stochastic process. Definitions and statements related to stochastic IM process are in \cref{sec:theory_of_stochastic_processes_IM_process}. Advantages of descriptions based on a stochastic process like IM process are summarized in \cref{sec:theory_of_stochastic_processes_adventages}.

\subsection{\label{sec:theory_of_stochastic_processes_generalities}General stochastic process}
\begin{definition}[Stochastic process]
\label{def:process}
A stochastic process is a family $X = \stochproc$ of random variables defined on the same probability space $\probabilityspace$ and, for fixed index $i$ in an index set $\indexset$, taking their values $\Xw[i=i]$ in given space $\stateset{i}$ which must be measurable with respect to some $\sigma$-algebra $\statesetmeasure{i}$ of admissible subsets. [This is \cref{process-stc-def:process} in \cite{Prochazka2022statphys_process_stc}.]
\end{definition}

The probability density functions $\rhoargSi{i}$ is given by \cref{process-stc-eq:rhoSinarg} in \cite{Prochazka2022statphys_process_stc}. It is functions of random variables $\Xw[i=i]$ and non-random variables $\XwNR[i=i]$ (the dependence on non-random variables will not be written explicitly in this \cref{sec:theory_of_stochastic_processes}). Number of states corresponding to given state space $\stateset{i}$ is denoted as $\Ni{i}$.

\begin{theorem}
\label{thm:rhoS_norm}
\letbeprocess{def:process} 
If $\stateset{i} \neq \emptyset$ then 
\begin{align} 
\int_{\Xw[i=i]} \rhoargSi{i} \dXw[i=i] &=1                                         \label{eq:rhoS_norm} \\
                          \dosargSi{i} &= \Ni{i} \rhoargSi{i}                      \label{eq:dosargi} \\
						        \Ni{i} &= \int_{\Xw[i=i]} \dosargSi{i} \dXw[i=i]   \label{eq:Ni_given_by_integration} 
\end{align}
for all $i \in \indexset$. [This is \cref{process-stc-thm:rhoS_norm} in \cite{Prochazka2022statphys_process_stc}.]
\end{theorem}

\begin{remark}[Density of states]
Quantity $\dosargSi{i}$ is called \emph{density of states} corresponding to given state space $\stateset{i}$ (it is defined in \cref{process-stc-sec:dos} in \cite{Prochazka2022statphys_process_stc}). 
\end{remark}

\subsection{\label{sec:theory_of_stochastic_processes_IM_process}Stochastic IM process}
The probability functions $\ProbTiarg{i}$ and probability density function $\rhoCiarg{i}$ mentioned in \cref{sec:problem_statement} are given by \cref{process-stc-def:process_sequence_notation_probT_rhoC} in \cite{Prochazka2022statphys_process_stc}.

\begin{assumption}[Possibility of no transition]
\label{as:state_prob_transition}
Probability $\ProbTiarg{i}$ may have values in the interval from 0 to 1, not necessarily in the whole interval. [This is \cref{process-stc-as:state_prob_transition} in \cite{Prochazka2022statphys_process_stc}.]
\end{assumption}

\begin{assumption}[Different state spaces]
\label{as:different_state_spaces}
		State spaces $\stateset{i}$ may or may not be the same for all $i \in \indexset$. [This is  \cref{process-stc-as:different_state_spaces} in \cite{Prochazka2022statphys_process_stc}.] 
\end{assumption}

\begin{assumption}[Independence of realizations]
\label{as:independent_realizations}
Realizations (outcomes) of an experiment are mutually independent. [This is \cref{process-stc-as:independent_realizations} in \cite{Prochazka2022statphys_process_stc}.] 
\end{assumption}

\begin{assumption}[memoryless (Markov) property] 
\label{as:markov_property} 
\letbeprocess{def:process} 
Let it have Markov property. I.e., roughly speaking, it means that probability of a future state of a system depends on its present state, but not on the past states in which the system was. [This is \cref{process-stc-as:markov_property} in \cite{Prochazka2022statphys_process_stc}.]
\end{assumption}

\newcommand{\refsasprocessSTCsequencefullPOL}{as:state_prob_transition,as:different_state_spaces,as:independent_realizations,as:markov_property}
\begin{definition}[IM process]
\label{def:process_STC}
\letbeprocess{def:process}
Let it satisfy \cref{\refsasprocessSTC}. This process is called independent and memoryless (IM) process. [This is \cref{process-stc-def:process_STC} in \cite{Prochazka2022statphys_process_stc}.]
\end{definition}
\begin{remark}
IM process is called state-transition-change (STC) process in \cite{Prochazka2022statphys_process_stc}. 
\end{remark}

\begin{definition}[IM sequence]
\label{def:process_STC_sequence}
\letbeprocess{def:process_STC} 
		Let the index set $\indexset$ be totally ordered sequence $(0, \dotsc, \ntrans)$ and $\Ni{0} \neq 0$. This process satisfies \cref{\refsasprocessSTCsequence}. [This is \cref{process-stc-def:process_STC_sequence} in \cite{Prochazka2022statphys_process_stc}.]
\end{definition}
\begin{theorem}[Transformation of density of states]
\label{thm:transition_dos}
\letbeprocess{def:process_STC_sequence}
 It holds
\begin{align}
        \dosargSi{i+1}
		= 
		\int_{\Xw[i=i]}
        \dosargSi{i}
        \ProbTiarg{i}
        \rhoCiarg{i}
		\dXw[i=i] \, . \label{eq:transition_dos}
\end{align}
		[This is \cref{process-stc-thm:transition_dos} in \cite{Prochazka2022statphys_process_stc}.]
\end{theorem}

\begin{remark}
\label{rmk:restrictions}
Stochastic processes discussed in literature typically assume $\ProbTiarg{i} = 1$. Some phenomena, such as absorption of a photon which may or may not be absorbed by an optical element, cannot be described under this assumption. \Cref{as:state_prob_transition} represent generalization of this assumption. 

It is also often assumed that states spaces $\stateset{i}$ are the same for two different indexes. This is also limiting assumptions in many cases. \Cref{as:different_state_spaces} removes this limitation.

The main assumptions under which one can derive \cref{thm:transition_dos} are \cref{as:independent_realizations,as:markov_property}. These two assumptions are often used and discussed in the context of stochastic processes. Not all stochastic processes, of course, are based on these two assumptions. However, relatively large and important set of phenomena corresponds to these two assumptions and they allow to derive the very general \cref{eq:transition_dos}. 
%
%
\end{remark}
\begin{theorem}
\letbeprocess{def:process_STC_sequence} 
It holds ($N_0 \neq 0$)
\begin{align}
		&\frac{N_{i+1}}{N_0} \rhoargSi{i+1}
		= 
		\int_{\Xw[i=i]}
        \frac{N_i}{N_0}  \rhoargSi{i} 
		\ProbTiarg{i}
		\rhoCiarg{i}
		\dXw[i=i]   \, .
\label{eq:dos_ip_over_N0_case_general} 
\end{align}
for $i \in (0,...,\ntrans-1)$.
\end{theorem}
\begin{proof}
\Cref{eq:transition_dos} divided by $N_0$ and \cref{eq:dosargi} imply \cref{eq:dos_ip_over_N0_case_general}.
\end{proof}

\begin{theorem}
\letbeprocess{def:process_STC_sequence} 
It holds ($N_0 \neq 0$)
\begin{align}
		\frac{N_{i+1}}{N_{0}} &=  \int_{\Xw[i=i]}
		\frac{N_{i}}{N_0}  \rhoargSi{i}
		\ProbTiarg{i}
		\dXw[i=i]  \label{eq:N_ip_over_N0_case_general}
\end{align}
for $i \in (0,...,\ntrans-1)$.
\end{theorem}
\begin{proof}
\Cref{process-stc-eq:transition_N_over_N0} in \cite{Prochazka2022statphys_process_stc} and \cref{eq:dosargi} imply \cref{eq:N_ip_over_N0_case_general}.
\end{proof}

\begin{remark}
\label{rmk:system_of_equations}
Let us assume that relative numbers of particles $\frac{N_{i+1}}{N_{0}}$ are measured as functions of some non-random variables (see, e.g., \cref{sec:measurement}). In this case \cref{eq:N_ip_over_N0_case_general,eq:dos_ip_over_N0_case_general} represent system of integral equations. It may or may not be possible to determine uniquelly the unknown functions $\rhoargSi{i}$, $\ProbTiarg{i}$ and $\rhoCiarg{i}$. It depends on dependence of $\frac{N_{i+1}}{N_{0}}$ on the non-random variables and involved assumptions used for discribtion of given system. If some other quantities than $\frac{N_{i+1}}{N_{0}}$ are measured then similar set of equations may be derived. The more experimental information is available, the better.
\end{remark}

\subsection{\label{sec:theory_of_stochastic_processes_adventages}Advantages of descriptions based on stochastic process like IM process}
Analysis of data based on IM process has several advantages (other stochastic processes have very similar or the same advantages):
\begin{enumerate}
\item{It allows to take into account and introduce only \emph{suitable random and non-random variables} characterizing given phenomenon. }
\item{It \emph{unifies} description of many phenomena from different fields of physics with the help of the theory of stochastic processes (the theory of probability), see \cite{Prochazka2022statphys_process_stc} in the case of IM process.}
\item{Analysis of data can be separated into two stages (see sect.~4.5 in \cite{Prochazka2022statphys_process_stc}). In the first stage one can try to determine the probability (density) functions on the basis of experimental data. In the second stage one can go into greater detail and try to explain (interpret) the probabilities in terms of some microscopic processes or some underlying functions (such as the Hamiltonian of the system, complex probability amplitude, wave function, etc.). The first stage does not require to introduce the underlying functions or to know various parameters of the microscopic processes (this information is often not a priori known and must be first determined on the basis of experimental data). This kind of additional more detailed knowledge about given system can be studied in the second stage. This separation of data analysis into the two stages in many cases allows to study and test various assumptions more effectively then trying to determine too many characteristics of given phenomena at once.}
\item{One can leverage terminology, techniques and results already known in the theory of stochastic processes (numerous stochastic processes are successfully used in many fields of research).}
\end{enumerate}

Using IM process, or other suitable stochastic process formulated with the help of the theory of stochastic processes, it is possible to overcome the difficulties and limitations existing in the contemporary statistical methods used in optics, see \cref{sec:contemprorary_theories_summary}.

\begin{remark}[Mendelian genetics]
\label{rmk:genetics}
As far as we know, similar separation of data analysis has not been systematically done in the context of polarization of light and optics in general. The separation has already helped enormously, e.g., in the field of genetics (biophysics) pioneered by Mendel in 1865 \cite{Mendel1865}. He performed many plant hybridization experiments. He was able to explain variability of observed properties of organisms in different generations by introducing dominant and recessive traits (characteristics). By introducing additional assumptions how these traits are inherited he was able to calculate probabilities of occurrences of organisms of given properties and in given generation (i.e., transition). The assumptions are now called laws of Mendelian genetics. The calculated probabilities agreed with the observed (measured) numbers. His analysis of data corresponded to the first stage of data analysis (i.e., determination of the probabilities without trying to explain them in greater detail). Genetic research is in many cases nowadays well in the second stage where determined dominant and recessive abstract traits are identified with real biophysical structures in an organism (a gene consisting of two alleles) and related processes are further studied in greater detail. The probabilistic model formulated by Mendel \cite{Mendel1865} allowed him to easily calculate the probabilities needed for comparison to data (formulation of the model itself and realization of all the required experiments, however, required surely far more effort and time). 

The law of independent assortment in Mendelian genetics is closely related to assumption of independence of outcomes of an experiment, and memoryless (Markov) property is assumed implicitly. Mendelian genetics is essentially described using IM process, even if it is not explicitly mentioned. Mendel used very special case of IM process as neither random nor non-random variable which would be necessary to determine from data was introduced in his probabilistic model. One can find other examples of IM processes implicitly used in literature. These cases correspond to relatively simple cases. Techniques developed for general IM process can help in some mathematically more complicated cases with the determination of the probabilities (the first stage) which is in general delicate task, see \cite{Prochazka2022statphys_process_stc} where some more complicated cases are mentioned, too. Explanation of the determined probabilities (the second stage) can be often studied later and separately, it can be also very demanding. The example from genetics (biophysics) clearly shows that some phenomena can be hardly understood if the corresponding analysis of data is not separated into the two stages and each of them studied separately (or quasi-separately). 
\end{remark}

We will use the strategy with separation of data analysis (systematically introduced in \cite{Prochazka2022statphys_process_stc}) into the two stages in the following in the context of optics to describe transmission of light through various optical elements. In the first stage of the analysis (description), it is not necessary to know and understand the microscopic structures of the optical elements or the detailed structure of photon.

\section{\label{sec:model}IM process and polarization of light}

Interaction of photons with various optical elements may be considered stochastic process. Therefore, one may ask how the theory of stochastic processes can help to describe this kind of phenomena. Discrete-index IM stochastic process, see \cref{sec:theory_of_stochastic_processes_IM_process}, is suitable for this purpose. Formulae describing transmission of light through sequence of polarization sensitive elements are derived in this section. 


Measured (relative) numbers of transmitted photons $N_{i}(\vec{\rotangle})/N_{0}$ characterize a transmission of light through sequence of optical elements, but they do not explain the phenomenon. If light beam consists of photons then we may ask how a single photon interacts with the optical element, what the probability (density) functions $\rhoSi{i}$, $\ProbTi{i}$ and $\rhoCi{i}$ characterizing the transmission are. 


We will derive main formulae corresponding to 3 different processes (cases) based on slightly different sets of assumptions. In \cref{sec:model_case_1} one can find formulae corresponding to sequence of quite general polarization sensitive elements (case 1). \Cref{sec:model_case_2} contains formulae valid for more special case under assumption that $\rhoCi{i}$ does not depend on random variables $\Xw[i=i]$ (case 2). In \cref{sec:model_case_3} one can find formulae valid in the case of even more special case assuming further that the elements in the sequence are identical (case 3).

\subsection{\label{sec:model_case_1}Case 1 - sequence of polarization sensitive elements}

\begin{definition}[Photon polarization angle]
\label{def:polarization}
Some types of polarization sensitive elements (such as linear polarizer or Faraday rotator) are sensitive to photon property called polarization and having meaning of an angle $\polarization[i=i]$. It characterizes the state of a photon when it passed through surface $\surface_i$.
\end{definition}
\begin{remark}
As to the \cref{def:polarization}, the photon polarization angle can specify the direction of a vector quantity characterizing the photon (spin, orientation vector, ...) projected onto the plane perpendicular to the direction of the photon velocity; or projection of direction of some photon oscillations onto the plane. The particular physical meaning of this variable is not important in the presented paper. In the following it is necessary to only know that it has meaning of an angle. 
\end{remark}
\begin{assumption}[Variables]
\label{as:Xw}
Let random variables $\Xw[i=i]$ characterizing $i$-th state of system (i.e., transmission of a photon through surface $\surface_i$) for given index $i \in \indexset$ where index set $\indexset$ is sequence $(0, \dotsc, \ntrans)$ be
\begin{align}
		\Xw[i=i] &= (\polarization[i=i])    \label{eq:Xw_def} ,
\end{align}
i.e., the photon polarization angles are random variables. Let non-random variables $\XwNR[i=i]$ characterizing $i$-th state of system be
\begin{align}
		\XwNR[i=i] &= (\rotangle[i=0], \dotsc, \rotangle[i=i-1] ) \, , \label{eq:XwNR_def}
\end{align}
i.e., the rotation angles of the polarization sensitive elements are non-random variables (parameters). $\ntrans+1$ random variables and $\ntrans$ non-random variables are used in total for description of the whole system.
\end{assumption}

\begin{assumption}[State spaces]
\label{as:state_space}
Let state space $\stateset{i}$ be a set of states represented by random variables $\Xw[i=i]$ and non-random variables $\XwNR[i=i]$. I.e., state space $\stateset{i}$ contains polarization states of photons when they pass through surface $\surface_i$ given fixed rotation angles of optical elements preceding $i$-th optical element.
\end{assumption}
\begin{remark}
\Cref{as:state_space} implies that the number of states of a system which were in a state in given state space $\stateset{i}$ is the same as the number of photons which passed through surface $\surface_i$.
\end{remark}

\begin{definition}
Functions introduced in \cref{process-stc-def:process_notation_dos} in \cite{Prochazka2022statphys_process_stc} can be written also as (see \cref{as:Xw})
\begin{alignat}{3}
		\dosargipol{i}            &= \dosargstateSi{i}     && \quad\quad\quad i=0,\dotsc,\ntrans         \label{eq:notation_dosargipolia}\\
		N_{i}(\vec{\rotangle})     &= \Ni{i}                && \quad\quad\quad i=0,\dotsc,\ntrans         \label{eq:notation_Ni} \\
		\rhoSargipol{i}            &= \rhoargSi{i}          && \quad\quad\quad i=0,\dotsc,\ntrans         \label{eq:notation_rhoSargipoli}
\end{alignat}
where the notation introduced in \cite{Prochazka2022statphys_process_stc} is on the right-hand sides of the equations. Dependence of the functions on left-hand side on non-random variables (i.e., on the rotation angles of the polarization sensitive elements) may or may not be written explicitly. 

Similarly, the probability of transition $\ProbTiarg{i}$ and the probability density function $\rhoCiarg{i}$ introduced by \cref{process-stc-def:process_sequence_notation_probT_rhoC} in \cite{Prochazka2022statphys_process_stc} can be written as
\begin{alignat}{3}
		\ProbTiargipol  &= \ProbTiarg{i}   && \quad i=0,\dotsc,\ntrans-1      \label{eq:ProbTiarg} \\
		\rhoCiargipol   &= \rhoCiarg{i}    && \quad i=0,\dotsc,\ntrans-1      \label{eq:rhoCiarg} \, .
\end{alignat}
\end{definition}

\begin{remark} 
In the case of transmission of light through sequence of polarization sensitive elements the transmission of a photon through $i$-th element is described by 3 functions $\rhoS[i=i]$, $\ProbTi{i}$ and $\rhoCi{i}$ corresponding to the 3 probabilistic effects mentioned in \cref{sec:problem_statement}:
\begin{enumerate}
\item{$\rhoSargipol{i}$ is probability density function characterizing distribution of polarization angles $\polarization[i=i]$ of incoming photons before an interaction with $i$-th element (i.e., polarization states of photons which passed through surface $\surface_i$). This function is normalized to 1 when integrated over all the possible polarization states $\polarization[i=i]$. It may be taken as dependent on the rotation angles $\vec{\rotangle}$ of the axes of all the elements in the sequence preceding the $i$-th element. By multiplying it by number of photons $N_{i}(\vec{\rotangle})$ one obtains density of photon polarization states $\dosargipol{i}$.}  
\item{$\ProbTiargipol$ is conditional probability that photon is transmitted through $i$-th element being rotated by angle $\rotangle[i=i]$ given the input photon polarization is $\polarization[i=i]$. Values of this probability function are in the interval from 0 to 1 (not necessarily in the whole interval).}
\item{$\rhoCiargipol$ is conditional probability density function characterizing change of input photon polarization $\polarization[i=i]$ to output photon polarization $\polarization[i=i+1]$ after the photon is transmitted through $i$-th element rotated by $\rotangle[i=i]$ given that the photon passed through the element. This function is normalized to 1 when integrated over all the possible values of the outgoing polarization states $\polarization[i=i+1]$ (independently on the value of the incoming photon polarization $\polarization[i=i]$ and the rotation $\rotangle[i=i]$).}
\end{enumerate}
The function $\rhoS[i=i]$ characterizes property of the photon beam before an interaction with the $i$-th element, and the functions $\ProbTi{i}$ and $\rhoCi{i}$ characterize interaction of a photon with the $i$-th element. 
\end{remark}
%
%

\newcommand{\refsasprocessPOLgeneralextra}{as:Xw,as:state_space}

\begin{definition}[Stochastic process - case 1] 
\label{def:process_polarizers_case_1}
Let $\stochproc$ be stochastic IM process given by \cref{def:process_STC_sequence} which satisfies \cref{\refsasprocessSTCsequencefullPOL}. 
Let it satisfy also \cref{\refsasprocessPOLgeneralextra}. 
\end{definition}
\begin{remark}
\label{rmk:main_assumptions}
Stochastic process given by \cref{def:process_polarizers_case_1} is IM process, i.e., it is based on the two main \cref{as:independent_realizations,as:markov_property}. \Cref{as:independent_realizations} concerning independence of outcomes of an experiment means independence of transmissions of individual photons through a sequence of optical elements. \Cref{as:markov_property} about memoryless (Markov) property means that transmission of a photon in $i$-th state through $(i+1)$-th optical elements depends on the $i$-th state and not on any of the states in which the photon was before.
\end{remark}

Several formulae derived in \cref{sec:theory_of_stochastic_processes_IM_process} or \cite{Prochazka2022statphys_process_stc} will be needed for analysis of data corresponding to transmission of light through a sequence of $M$ polarization sensitive devices using stochastic processes given by \cref{def:process_polarizers_case_1}. The formulae can be rewritten using the notation introduced above. 

\begin{theorem}
\letbeprocess{def:process_polarizers_case_1} 
It holds
\begin{alignat}{3} 
		1 \geq& \frac{N_{1}(\vec{\rotangle})}{N_{0}} \geq \frac{N_{2}(\vec{\rotangle})}{N_{0}} \geq ... \geq \frac{N_{\ntrans}(\vec{\rotangle})}{N_{0}} && \geq 0  \, .  \label{eq:Ni_to_N0_interval_chain_simplified_pol} 
\end{alignat}
\end{theorem}
\begin{proof}
It follows from \cref{process-stc-eq:Ni_to_N0_interval_chain} in \cite{Prochazka2022statphys_process_stc} and \cref{as:Xw}.
\end{proof}
\begin{theorem}
\letbeprocess{def:process_polarizers_case_1} 
It holds
\begin{alignat}{3}
		\int_{\polarization[i=i]} 
		\rhoSargipol{i}  
		\text{d} 
		\polarization[i=i]      
		&&= 1 & \qquad i=0,...,M  \label{eq:rhoS_norm_pol} \, .
\end{alignat}
\end{theorem}
\begin{proof}
This normalization condition can be derived using \cref{process-stc-eq:rhoS_norm} in \cite{Prochazka2022statphys_process_stc} and \cref{eq:notation_rhoSargipoli}. 
\end{proof}
\begin{theorem}
\letbeprocess{def:process_polarizers_case_1} 
It holds
\begin{alignat}{3}
		\int_{\polarization[i=i+1]} \rhoCiargipol  \text{d} \polarization[i=i+1]  &&= 1 &  \, .  \label{eq:rhoargstateC_norm_case_1}
\end{alignat}
\end{theorem}
\begin{proof}
It follows from \cref{process-stc-eq:rhoargstateC_norm} in \cite{Prochazka2022statphys_process_stc} and \cref{as:Xw}.
\end{proof}
\begin{remark}
The normalization condition \refp{eq:rhoargstateC_norm_case_1} holds for arbitrary value of $\polarization[i=i]$.
\end{remark}
\begin{theorem}
\letbeprocess{def:process_polarizers_case_1} 
It holds
\begin{alignat}{3} 
\dosargipol{i} &= N_i(\vec{\rotangle})  \rhoSargipol{i} && \qquad i=0,...,M      \, . \label{eq:def_flux_EventSXw_i_rho_simplified_pol} 
\end{alignat}
\end{theorem}
\begin{proof}
It follows from \cref{eq:dosargi} and \cref{eq:notation_dosargipolia,eq:notation_Ni,eq:notation_rhoSargipoli}.
\end{proof}
\begin{theorem}
\letbeprocess{def:process_polarizers_case_1} 
It holds
\begin{alignat}{3} 
		N_i(\vec{\rotangle})  &= \int_{\polarization[i=i]} \dosargipol{i} \text{d} \polarization[i=i] && \qquad i=0,...,M \, .  \label{eq:Ni_given_by_integration_simplified_pol}
\end{alignat}
\end{theorem}
\begin{proof}
\Cref{eq:Ni_given_by_integration,eq:notation_dosargipolia,eq:notation_Ni} imply \cref{eq:Ni_given_by_integration_simplified_pol}. 
\end{proof}
\begin{remark}
The density of states $\dosargipol{i}$ essentially defined by \cref{eq:def_flux_EventSXw_i_rho_simplified_pol} represents spectrum of values of polarization angles of photons which passed through surface $\surface_i$; the spectrum is normalized to the number of photons $N_i$ which in total passed through surface $\surface_i$, see \cref{eq:Ni_given_by_integration_simplified_pol}.
\end{remark}

\begin{theorem}
\letbeprocess{def:process_polarizers_case_1} 
It holds ($N_0 \neq 0$)
\begin{alignat}{3}
		\frac{\dosargipol{i}}{N_0} &= \frac{N_i(\vec{\rotangle})}{N_0}  \rhoSargipol{i}  && \qquad i=0,...,M \label{eq:def_rel_flux_EventSXw_i_rho_simplified} 
\end{alignat}
\end{theorem} 
\begin{proof}
\Cref{eq:def_flux_EventSXw_i_rho_simplified_pol} can be divided by $N_0$.
\end{proof}
\begin{theorem}
\letbeprocess{def:process_polarizers_case_1} 
It holds ($N_0 \neq 0$)
\begin{alignat}{3}
				\frac{N_i(\vec{\rotangle})}{N_0}  &= \int_{\polarization[i=i]} \frac{\dosargipol{i}}{N_0} \text{d} \polarization[i=i] && \qquad i=0,...,M \, .  \label{eq:Ni_over_N0_given_by_integration_simplifed}
\end{alignat}
\end{theorem}
\begin{proof}
\Cref{eq:Ni_given_by_integration_simplified_pol} can be divided by $N_0$.
\end{proof}
\begin{theorem}
\letbeprocess{def:process_polarizers_case_1} 
It holds
\begin{align}
		&\frac{N_{i+1}(\vec{\rotangle})}{N_0}\rhoSargipol{i+1} 
		= \nonumber  \\
		&\qquad\qquad
		\int_{\polarization[i=i]}
        \frac{N_i(\vec{\rotangle})}{N_0}  \rhoSargipol{i} 
        \ProbTiargipol
\rhoCiargipol
		\text{d} \polarization[i=i]   \, .
\label{eq:dos_ip_over_N0_case_1} 
\end{align}
for $i \in (0,...,\ntrans-1)$.
\end{theorem}
\begin{proof}
\Cref{eq:dos_ip_over_N0_case_general} can be rewritten using \cref{eq:notation_rhoSargipoli,eq:ProbTiarg,eq:rhoCiarg}.
\end{proof}
\begin{theorem}
\letbeprocess{def:process_polarizers_case_1} 
It holds
\begin{align}
		\frac{N_{i+1}(\vec{\rotangle})}{N_{0}} &=  \int_{\polarization[i=i]}
		\frac{N_{i}(\vec{\rotangle})}{N_0}  \rhoSargipol{i}
        \ProbTiargipol
		\text{d} \polarization[i=i]  \label{eq:N_ip_over_N0_case_1}
\end{align}
for $i \in (0,...,\ntrans-1)$.
\end{theorem}
\begin{proof}
\Cref{eq:N_ip_over_N0_case_general} can be rewritten using \cref{eq:notation_rhoSargipoli,eq:ProbTiarg}.
\end{proof}
%
%
%
\begin{remark}
\Cref{eq:dos_ip_over_N0_case_1} is of key importance. If the initial probability density function $\rhoSpol{0}$ and the functions $\ProbTi{i}$ and $\rhoCi{i}$ are given for all $i=0,...,M-1$ then the formula \refp{eq:dos_ip_over_N0_case_1} may be used iteratively to calculate $\frac{N_i(\vec{\rotangle})}{N_0}  \rhoSargipol{i}$ for all $i=1,...,M$ (i.e., behind each polarization sensitive element). Let us emphasize that an outgoing photon after being transmitted through a polarization sensitive element can have different value of polarization, and it becomes incoming photon interacting with the next element in the sequence. The photon polarization angles corresponding to the surfaces $\surface_i$ are, therefore, distinguished by index $i$ in \cref{eq:dos_ip_over_N0_case_1}. 
\end{remark}

\subsection{\label{sec:model_case_2}Case 2 - sequence of polarization sensitive elements and independence of $\rhoCiargipol$ on $\polarization[i=i]$}
\begin{definition} 
If probability density function $\rhoCiargipol$ does not depend on polarization angle of incoming photon then the function may be denoted as $\rhoCindepargipol$. 
\end{definition}
\begin{assumption}[Independence of $\rhoCiargipol$ on $\protect{\polarization[i=i]}$]
\label{as:rhoC_independent_of_polarization_in_pol}
Let probability density function $\rhoCiargipol$ do not depend on $\polarization[i=i]$, i.e., it holds 
\begin{align}
    \rhoCindepargipol      &= \rhoCiargipol       & i&=0,...,\ntrans-1 \, .  \label{eq:notation_rhoCindepargipol} 
\end{align}
\end{assumption}
\begin{remark}
\label{rmk:rhoC_independent_of_polarization_in_pol}
\Cref{as:rhoC_independent_of_polarization_in_pol} is equivalent to \cref{process-stc-as:rhoC_independent_of_xiwin} in \cite{Prochazka2022statphys_process_stc} and \cref{as:Xw}.
\end{remark}
\newcommand{\refsasprocessPOLcaseoneimpliedextra}{as:rhoC_independent_of_polarization_in_pol}
\begin{definition}[Stochastic process - case 2] 
\label{def:process_polarizers_case_2}
Let $\stochproc$ be stochastic process given by \cref{def:process_polarizers_case_1} and let it satisfy also \cref{\refsasprocessPOLcaseoneimpliedextra}.
\end{definition}
\begin{remark}
Stochastic process given by \cref{def:process_polarizers_case_2} is process given by \cref{process-stc-def:process_STC_sequence_case_one} in \cite{Prochazka2022statphys_process_stc} which satisfies \cref{\refsasprocessSTCsequencecaseonefull} in \cite{Prochazka2022statphys_process_stc}, and \cref{\refsasprocessPOLgeneralextra}. It implies that \cref{\refsasprocessPOLcaseoneimpliedextra} holds (see \cref{rmk:rhoC_independent_of_polarization_in_pol}).
\end{remark}

\begin{theorem}
\label{thm:dos_ip_over_N0_case_2}
\letbeprocess{def:process_polarizers_case_2} 
It holds
\begin{align}
		&\frac{N_{i+1}(\vec{\rotangle})}{N_0}\rhoSargipol{i+1} 
		=\nonumber \\
		&\qquad\qquad
		\rhoCindepargipol
        \left[  
		\int_{\polarization[i=0]}
        \rhoSargipol{0}
		\ProbTi{0}(\polarization[i=0], \rotangle[i=0]) 
		\text{d} \polarization[i=0] \right] \nonumber \\
		&\qquad\qquad\prod_{j = 1}^{i} \left[ \int_{\polarization[i=j]}
        \rhoCindep[i=j-1](\polarization[i=j], \rotangle[i=j-1])
		\ProbTi{j}(\polarization[i=j], \rotangle[i=j]) 
		\text{d} \polarization[i=j] \right] 
		\, .
\label{eq:dos_ip_over_N0_case_2}
\end{align}
where $\prod_{j = 1}^{i} \left[...\right]=1$ if $i\!=\!0$. 
\end{theorem}
\begin{proof}
\Cref{process-stc-eq:transition_N_over_N0_simplified_rho_independent_on_xi_and_identical} in \cite{Prochazka2022statphys_process_stc} can be rewritten using \cref{as:Xw,as:rhoC_independent_of_polarization_in_pol,eq:notation_rhoCindepargipol} to obtain \cref{eq:dos_ip_over_N0_case_2}.
\end{proof}
\begin{remark}
\label{rmk:same_shapes}
The functions $\frac{N_{i+1}(\vec{\rotangle})}{N_0}\rhoSargipol{i+1}$ of $\polarization[i=i+1]$ given by \cref{eq:dos_ip_over_N0_case_2} have the same shapes for all $i=0,...,\ntrans-1$; the functions differ only in normalizations. 
\end{remark}
\begin{remark}
\label{rmk:simpler_numeric}
\Cref{eq:dos_ip_over_N0_case_2} contains multiplication of integrals, while the integrals in \cref{eq:dos_ip_over_N0_case_1} are calculated iteratively (recursively). Therefore, the functions $\frac{N_{i+1}(\vec{\rotangle})}{N_0}\rhoSargipol{i+1}$ are easier to calculate numerically using \cref{eq:dos_ip_over_N0_case_2} than using \cref{eq:dos_ip_over_N0_case_1}. However, this comes at the cost of loss of generality as \cref{as:rhoC_independent_of_polarization_in_pol} must be introduced.
\end{remark}

\subsection{\label{sec:model_case_3}Case 3 - sequence of identical polarization sensitive elements and independence of $\rhoCiargipol$ on $\polarization[i=i]$}
\begin{definition}[Probability of transmission $\ProbTidentityargpol$]
If the probabilities $\ProbTiargipol$ are the same for all $i=0,1,...,\ntrans-1$, then the probability of transmission of photon through a polarization sensitive element can be denoted as $\ProbTidentityargpol$. It depends on the polarization of the incoming photon $\polarization[i=\text{in}]$ and the rotation angle $\rotangle$ of the element.
\end{definition}
\begin{assumption}[Identical probabilities $\ProbTiargipol$]
\label{as:ProbT_identity_pol}
The probability functions $\ProbTiargipol$ are the same for all $i=0,1,...,\ntrans-1$, i.e., for all given polarization sensitive elements. I.e., it holds
\begin{align}
\ProbTargiexpandedpol              &= \ProbTiargipol        & i&=0,...,\ntrans-1 \, . \label{eq:notation_ProbTargpol}
\end{align}
\end{assumption}
\begin{definition}
If the probability density functions $\rhoCiargipol$ are the same for all $i=0,1,...,\ntrans-1$, then the function can be denoted as $\rhoCargpol$. It depends on the polarization of the outgoing photon $\polarization[i=\text{out}]$ and the rotation of given optical element $\rotangle$. 
\end{definition}
\begin{assumption}[Identical functions $\rhoCiargipol$] 
\label{as:rhoC_identity_pol}
Let the probability density functions $\rhoCiargipol$ be the same for all polarization sensitive elements. I.e., it holds (\mbox{$i=0,...,\ntrans-1$})
\begin{align} 
		\rhoCargexpandedpol            &= \rhoCiargipol         & & \, . \label{eq:notation_rhoCargpol}
\end{align}
\end{assumption}
\begin{remark}
\label{rmk:as_rho_C_identity_pol} 
\Cref{as:rhoC_identity_pol} is equivalent to \cref{process-stc-as:rhoC_identity} in \cite{Prochazka2022statphys_process_stc} and \cref{as:Xw}. 
\end{remark}
\begin{assumption}
\label{as:identical_elements}
Let all the polarization sensitive elements in given sequence be identical.
\end{assumption}
\begin{remark}
\label{rmk:as_identical_elements}
\Cref{as:identical_elements} implies \cref{as:ProbT_identity_pol,as:rhoC_identity_pol}.
\end{remark}

\begin{definition}
Let the probability density function corresponding to both the \cref{as:rhoC_independent_of_polarization_in_pol,as:rhoC_identity_pol} be denoted as $\rhoCindepidentityargpol$, i.e.,
\begin{align}
		\rhoCindepidentityargiexpandedpol     &= \rhoCiargipol         & i&=0,...,\ntrans-1 \, . \label{eq:notation_rhoCindepidentityargipol}
\end{align}
\end{definition}

\newcommand{\refsasprocessPOLcasetwoimpliedextra}{as:rhoC_identity_pol}
\newcommand{\refsasprocessPOLcasetwoextra}{as:ProbT_identity_pol}
\begin{definition}[Stochastic process - case 3] 
\label{def:process_polarizers_case_3}
Let $\stochproc$ be stochastic process given by \cref{def:process_polarizers_case_2} and satisfying also \cref{\refsasprocessPOLcasetwoimpliedextra,\refsasprocessPOLcasetwoextra}. 
\end{definition}
\begin{remark}
Stochastic process given by \cref{def:process_polarizers_case_3} is stochastic process given by \cref{process-stc-def:process_STC_sequence_case_two} in \cite{Prochazka2022statphys_process_stc} which satisfies \cref{\refsasprocessSTCsequencecasetwofull} in \cite{Prochazka2022statphys_process_stc}, \cref{\refsasprocessPOLgeneralextra} and \cref{\refsasprocessPOLcasetwoextra}. It implies that \cref{\refsasprocessPOLcasetwoimpliedextra} holds (see \cref{rmk:as_rho_C_identity_pol}) and that \cref{as:identical_elements} holds (see \cref{rmk:as_identical_elements}). 
\end{remark}
Let as further assume that all polarization sensitive elements in given sequence are identical (see \cref{as:identical_elements}), i.e., let us assume that \cref{as:ProbT_identity_pol,as:rhoC_identity_pol} are satisfied.  
\begin{theorem}
\letbeprocess{def:process_polarizers_case_3} 
It holds
\begin{alignat}{3}
		\int_{\polarization[i=\text{out}]} \rhoCargpol   \text{d} \polarization[i=\text{out}]  &&= 1 &  \, .  \label{eq:rhoargstateC_norm}
\end{alignat}
\end{theorem}
\begin{proof}
It follows from \cref{eq:rhoargstateC_norm_case_1} and \cref{eq:notation_rhoCargpol}.
\end{proof}

\begin{theorem}
\letbeprocess{def:process_polarizers_case_3} 
It holds
\begin{align}
		&\frac{N_{i+1}(\vec{\rotangle})}{N_0}\rhoSargipol{i+1} 
		= \nonumber  \\
		&
		\int_{\polarization[i=i]}
        \frac{N_i(\vec{\rotangle})}{N_0}  \rhoSargipol{i} 
        \ProbTargipol
\rhoCargipol
		\text{d} \polarization[i=i]  
\label{eq:dos_ip_over_N0_case_3}
\end{align}
for $i \in (0,...,\ntrans-1)$. 
\end{theorem}
\begin{proof}
\Cref{eq:dos_ip_over_N0_case_1} can be rewritten using \cref{eq:notation_ProbTargpol,eq:notation_rhoCargpol}. The two functions $\ProbTi{}$ and $\rhoCi{}$ are the same for all the \emph{identical} polarization sensitive elements in the sequence, but it must be correctly integrated over their arguments, see \cref{eq:dos_ip_over_N0_case_3}.
\end{proof}
\begin{theorem}
\letbeprocess{def:process_polarizers_case_3} 
It holds
\begin{align}
		\frac{N_{i+1}(\vec{\rotangle})}{N_{0}} &=  \int_{\polarization[i=i]}
		\frac{N_{i}(\vec{\rotangle})}{N_0}  \rhoSargipol{i}
        \ProbTargipol
		\text{d} \polarization[i=i]  \label{eq:N_ip_over_N0_case_3}
\end{align}
for $i \in (0,...,\ntrans-1)$. 
\end{theorem}
\begin{proof}
\Cref{eq:N_ip_over_N0_case_1} can be rewritten using \cref{eq:notation_ProbTargpol}. The function $\ProbTi{}$ is the same for all the \emph{identical} polarization sensitive elements in the sequence, but it must be correctly integrated over its arguments, see \cref{eq:N_ip_over_N0_case_3}.
\end{proof}

\begin{theorem}
\letbeprocess{def:process_polarizers_case_3} 
It holds
\begin{align}
		&\frac{N_{i+1}(\vec{\rotangle})}{N_0}\rhoSargipol{i+1} 
		=\nonumber \\
		&\qquad\qquad
		\rhoCindep(\polarization[i=i+1], \rotangle[i=i]) 
        \left[  
		\int_{\polarization[i=0]}
        \rhoSargipol{0}
		\ProbTi{}(\polarization[i=0], \rotangle[i=0]) 
		\text{d} \polarization[i=0] \right] \nonumber \\
		&\qquad\qquad\prod_{j = 1}^{i} \left[ \int_{\polarization[i=j]}
        \rhoCindep(\polarization[i=j], \rotangle[i=j-1])
		\ProbTi{}(\polarization[i=j], \rotangle[i=j]) 
		\text{d} \polarization[i=j] \right] 
		\, .
\label{eq:dos_ip_over_N0_case_3_simplified}
\end{align}
\end{theorem}
\begin{proof}
\Cref{eq:dos_ip_over_N0_case_2} can be simplified using \cref{eq:notation_ProbTargpol,eq:notation_rhoCargpol,eq:notation_rhoCindepidentityargipol}.
\end{proof}
\begin{remark}
\label{rmk:simpler_numeric_2}
Numerical calculation of $\frac{N_{i+1}(\vec{\rotangle})}{N_0}\rhoSargipol{i+1}$ using \cref{eq:dos_ip_over_N0_case_3_simplified} is simpler than using \cref{eq:dos_ip_over_N0_case_3}, see also \cref{rmk:simpler_numeric}.
\end{remark}
\begin{remark}
\label{rmk:3_functions_to_be_parameterized}
\Cref{eq:N_ip_over_N0_case_3} and \cref{eq:dos_ip_over_N0_case_3} (or \cref{eq:dos_ip_over_N0_case_3_simplified}) allow to calculate the relative number of transmitted photons $N_{i}(\vec{\rotangle})/N_{0}$ for any $i=1,...,M$ according to \cref{eq:N_ip_over_N0_case_3} if 3 functions are known: $\rhoSargipol{0}$, $\ProbTidentityargpol$ and $\rhoCargpol$. These unknown functions may be parameterized and determined on the basis of measured (input) values of $N_{i}(\vec{\rotangle})/N_{0}$.
\end{remark}

\section{\label{sec:definitions}Definitions of types of optical elements and properties of light (beam)}
Given optical element (sample) transmits light in specific way. Standard way of classifying optical elements is based on measuring properties of photon beam before and after an interaction with given optical element (or using a sequence of elements). In \cref{sec:definition_optical_elements_types_transmission} one can find definitions of polarized and unpolarized beam and definitions of two types of polarization sensitive devices: linear polarizer and Faraday rotator.

Another way how to define various types of optical elements and properties of light beam is using probability (density) functions $\rhoSpol{}$, $\ProbTi{}$ and $\rhoCi{}$. This method is discussed in \cref{sec:definition_optical_elements_types_probability}. The first function characterizes property of photon beam and can be, therefore, used for definition of (un)polarized beam, see \cref{sec:definition_polarized_light_prob}. The other two functions can be used for definitions of various types of optical elements (such as linear polarizer and Faraday rotator), see \cref{sec:definition_optical_elements_prob}. Definitions concerning interaction of light (beam) with polarization sensitive elements are in \cref{sec:definition_interaction_prob}.

Given optical element defined (identified) with the help of the first approach should be equivalently and consistently defined using the second approach. Determination of functions $\ProbTi{}$ and $\rhoCi{}$ on the basis of experimental data basically requires experimental data needed to define the element using the first approach. Some of the definitions bellow can be specified more precisely, if needed. For our purposes it is sufficient to show mainly the basic idea behind each definition. 

\subsection{\label{sec:definition_optical_elements_types_transmission}Definitions based on measured properties of transmitted beam}
\subsubsection{\label{sec:definition_polarized_light}Definitions of properties of light (beam)}
\begin{definition}[Unpolarized light (beam)]
\label{def:unpolarized_light}
If a photon beam is transmitted through a polarization sensitive element and the number of transmitted photons does not depend on the rotation of the element then the light is unpolarized.
\end{definition}

\begin{definition}[Polarized light (beam)]
\label{def:polarized_light}
If light is not unpolarized, see \cref{def:unpolarized_light}, then it is polarized.
\end{definition}

\subsubsection{\label{sec:definition_optical_elements}Definitions of different types of optical elements}
\begin{definition}[Ideal identical linear polarizers]
\label{def:ideal_identical_linear_polarizers}
Ideal identical linear polarizers are polarization sensitive elements which transmit photons according to \cref{eq:N_i_over_N_0_ideal_polarizer} when initial photon beam is unpolarized (i.e., they behave according to Malus's law).
\end{definition}

\begin{definition}[(Real) linear polarizer]
An optical element transmitting light similarly as ideal linear polarizer (see also \cref{sec:real_polarizers}).
\end{definition}

\begin{definition}[Faraday rotator]
\label{def:faraday_rotator}
Let us consider unpolarized beam passing through sequence of a linear polarizer, an optical element and another linear polarizer. Number of photons transmitted through this sequence, divided by initial number of photons, can be measured as a function of the rotation of the second linear polarizer (the first one having fixed angle of rotation). This quantity can be measured with and without magnetic-field applied to the unknown optical element. If the two measured quantities have the same dependence on the angle of the rotation, but are shifted by an angle, and the shift depends on the parallel component of the magnetic field, then the unknown element is Faraday rotator (also called Faraday-effect based device).
\end{definition}
\begin{remark}
Measurement of this type is closely related to measurement of the Verdet constant and can be found in \cite{Vojna2019}. The angular shift is clearly visible in fig.~3 in \cite{Vojna2019}.
\end{remark}

\subsection{\label{sec:definition_optical_elements_types_probability}Definitions based on probability (density) functions}
\subsubsection{\label{sec:definition_polarized_light_prob}Definitions of properties of light (beam)}
Probability density function $\rhoSpol{}(\polarization)$ in dependence on polarization angle $\polarization$ characterizes distribution of polarization states of photons when they pass through given control surface. 

\begin{definition}[Unpolarized light (beam)]
\label{def:unpolarized_light_prob}
If $\rhoSpol{}(\polarization)$ does not depend on $\polarization$ then the photon beam passing through the control surface is unpolarized. I.e., photon polarization angles $\polarization$ of the photons are distributed uniformly and the distribution is normalized to 1 (see \cref{eq:rhoS_norm_pol}), it implies
\begin{equation}
		\rhoSpol{}(\polarization) = \frac{1}{2\pi} \, . \label{eq:rhoS_uniform}
\end{equation}
(full angle in radians is $2\pi$). 
\end{definition}

\begin{definition}[Polarized light (beam)]
\label{def:polarized_light_prob}
If $\rhoSpol{}(\polarization)$ depends on $\polarization$ then the photon beam passing through the control surface is polarized.
\end{definition}
\begin{remark}
\Cref{def:unpolarized_light_prob,def:polarized_light_prob} imply that if a photon beam passing through a control surface is not polarized then it is unpolarized (and vice versa).
\end{remark}

\subsubsection{\label{sec:definition_optical_elements_prob}Definitions of different types of optical elements}
Given optical element is characterized by functions $\ProbTi{}$ and $\rhoCi{}$ ($\ProbTidentityargpol$ and $\rhoCargpol$ in the case of polarization sensitive element). They can be determined on the basis of experimental data. These two functions can be, therefore, used for distinguishing various types of optical elements. Let us define, e.g., linear polarizer and Faraday rotator. 

\begin{definition}[Linear polarizer]
\label{def:linear_polarizer_prob}
A linear polarizer is an optical element which may or may not change polarization angle $\polarization[i=\text{in}]$ of an incoming photon such that the possible directions of polarization of the outgoing photon (specified by the polarization angle $\polarization[i=\text{out}]$) are predominantly parallel to an axis (called axis of the linear polarizer). I.e., probability density function $\rhoCargpol$ corresponding to this optical element has a peak at the position of the axis for any fixed value of $\polarization[i=\text{in}]$ and rotation of the axis $\rotangle$. If the outgoing photons have all the same value of $\polarization[i=\text{out}]$ then $\rhoCargpol$ is given by a delta function. One may also assume that probability function $\ProbTidentityargpol$ corresponding to this element depends on $\polarization[i=\text{in}]$ and the rotation of the axis \rotangle.
\end{definition}

\begin{definition}[Faraday rotator]
\label{def:faraday_rotator_prob}
A Faraday rotator (or Faraday-effect-based device) is an optical element having probability function $\ProbTidentityargpol$ which does not depend on polarization state of an incoming photon $\polarization[i=\text{in}]$ (it depends on parameters such as the length of the medium, its temperature, component of magnetic field applied to it and being parallel to the direction of the beam, ...). Moreover, the difference of polarization angle of an outgoing photon $\polarization[i=\text{out}]$ and the polarization angle $\polarization[i=\text{in}]$ of the incoming photon is the same (resp.~roughly the same) for each transmitted photon independently on the value of $\polarization[i=\text{in}]$. I.e., probability density function $\rhoCargpol$ as a function of the difference $\polarization[i=\text{in}]$ - $\polarization[i=\text{out}]$ is given by a delta function (resp.~by a function having a significant peak). 
\end{definition}



\subsubsection{\label{sec:definition_interaction_prob}Definitions of interaction of light with optical elements}
\begin{definition}[Polarizing and depolarizing transmission]
\label{def:polarizing_effect}
If $\rhoSpol{i+1}$ differs from uniform distribution more than $\rhoSpol{i}$ then the transmission is polarizing. If $\rhoSargipol{i+1}$ differs from uniform distribution less than $\rhoSargipol{i}$ then the transmission is depolarizing. 
\end{definition}
\begin{remark}
The difference of the two distributions in \cref{def:polarizing_effect} can be quantified, e.g., with the help of the second moments of the distributions. If the second moment of $\rhoSargipol{i+1}$ is lower (resp.~higher) than the second moment of $\rhoSargipol{i}$, then the transmission is polarizing (resp.~depolarizing). Or another rule can be used, if one of the moments is not finite. The second moments of the distributions may not be convenient characterization of "(de)polarization" if one of the distribution has more than one significant peak.
\end{remark}
\begin{remark}
$\frac{N_{i+1}(\vec{\rotangle})}{N_0}\rhoSargipol{i+1}$ characterizing density of photon polarization states behind $i$-th polarization sensitive element can be calculated on the basis of 3 functions $\frac{N_i(\vec{\rotangle})}{N_0}\rhoSargipol{i}$, $\ProbTiargipol$ and $\rhoCiargipol$, see \cref{eq:dos_ip_over_N0_case_1}. The output density of states depends on the input density of states and the properties of the element. Therefore, it is possible that an element can have polarizing or depolarizing effect on the beam in dependence on the input probability density function $\rhoSargipol{i}$. These two effects can be distinguished by compering functions $\rhoSargipol{i}$ and $\rhoSargipol{i+1}$. 
\end{remark}

\section{\label{sec:data_analysis}Example of analysis of data using IM process - 3 polarizers experiment}
In the previous sections the probabilistic and corpuscular theoretical description suitable for describing optical phenomena, especially those related to polarization, has been explained. In the following an example data of relative photon numbers $N_{i}(\vec{\rotangle})/N_{0}$ ($i=1,2,3$) measured as explained in \cref{sec:measurement_general} will be analyzed with the help of stochastic IM process adapted for description of polarization of light, see \cref{sec:model}. 

Choice of data sample is discussed in \cref{sec:data_example}. The analysis will be done with the help of the general guidelines summarized in \cref{process-stc-sec:data_analysis_guidelines} in \cite{Prochazka2022statphys_process_stc}. Similar analysis has never been done in the context of optics until now. We will be, therefore, interested mainly in the concept and possibilities of this new kind of analysis. We will not focus too much on numerical details. The analysis has helped to identify several open questions related to determination of functions $\ProbTi{}$ and $\rhoCi{}$ characterizing linear polarizers, they are discussed in \cref{sec:data_analysis_open_questions}.

\subsection{\label{sec:data_example}Example data - 3 ideal identical linear polarizers}
Relative photon numbers $N_{i}(\vec{\rotangle})/N_{0}$ corresponding to transmission of light through one and two linear polarizers in dependence on their rotation angles are commonly measured (e.g., to compare the measured intensities with the Malus's law). However, it seems that there are no publicly available experimental data of the relative beam intensities (number of transmitted photons) measured behind one, two and three linear polarizers in dependence on the rotations of the polarizers. In \cite{Krasa1993} one can find interesting experimental results concerning transmission of light through three linear polarizers, but not the measured values of all the relative numbers of transmitted photons $N_{i}(\vec{\rotangle})/N_{0}$ needed for our analysis.

Therefore, let us take (for the sake of simplicity) the dependences given by \cref{eq:N_i_over_N_0_ideal_polarizer} for $\ntrans=3$ corresponding to sequence of 3 ideal identical linear polarizers as an input for our analysis of data. It will be assumed that photons in the beam have the same energy, i.e., that \cref{eq:relation_N_to_I_i_to_0} holds. These dependences will be used in our further considerations as an example of measured data. In the following we will not focus on possible differences existing between the relative photon numbers corresponding to real and ideal polarizers (see \cref{sec:real_polarizers}). This well-defined example of measured data will be analyzed with the aim to determine probabilistic (statistical) characteristics of a photon transmitted through the sequence of linear polarizers.

\subsection{\label{sec:analysis_initial_exploration}Initial exploration of data}
The data in \cref{sec:data_example} correspond to a sequence of identical polarization sensitive elements (one can assume that the ideal polarizers are identical). Therefore, one can try to describe them with the help of stochastic process given by \cref{def:process_polarizers_case_3}, see also \cref{sec:model_case_3} for key formulae corresponding to this process. 

Experimental data often reveal \emph{symmetries}. This is also the case of the relative photon numbers $N_{i}(\vec{\rotangle})/N_{0}$ ($i=1,..,M$) in dependence on the orientation of the axes of the linear polarizers given by \cref{eq:N_i_over_N_0_ideal_polarizer}. The ratio $N_{1}(\vec{\rotangle})/N_{0}$ does not depend on the orientation of the axis of the first polarizer; and the ratios $N_{2}(\vec{\rotangle})/N_{0}$ and $N_{3}(\vec{\rotangle})/N_{0}$ are periodic due to the cosine-squared function in \cref{eq:N_i_over_N_0_igt1_ideal_polarizer}. One may define the following function
\begin{subnumcases}{\sym(x) = \label{eq:sym}}
\pi - y &\text{if $\frac{\pi}{2} < y <= \pi$ }\\
y       &\text{otherwise}
\end{subnumcases}
where $y = (\abs{x} \mod \pi )$. This function will help to reflect the symmetries. The function is even ($\sym(x)=\sym(-x)$) and is plotted in \cref{fig:sym}.
\begin{figure}
\captionsetup{width=0.95\columnwidth}
\centering
\includegraphics*[width=\columnwidth]{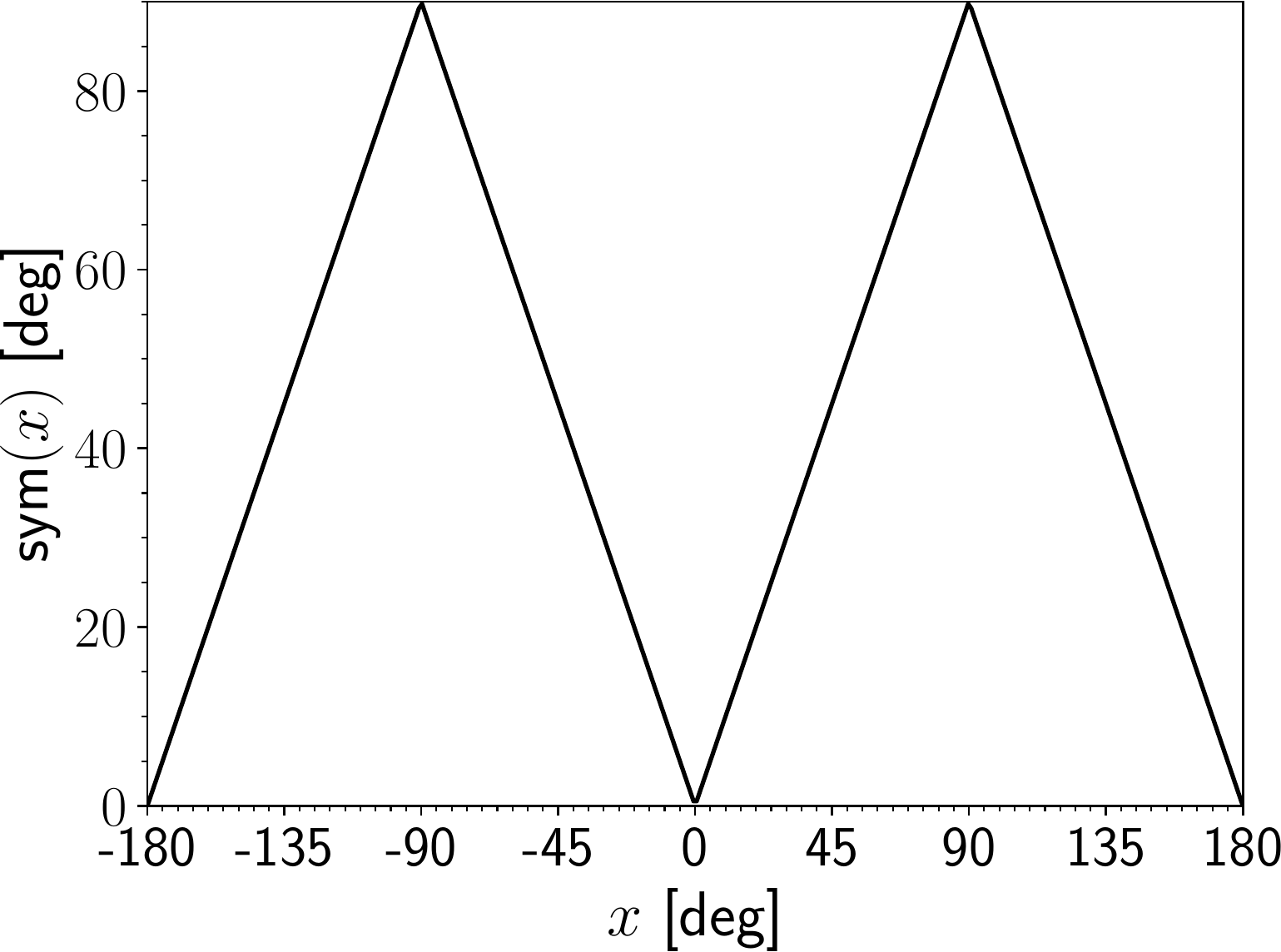}
		\caption{\label{fig:sym} The function $\sym(x)$ given by \cref{eq:sym} in the interval from -$\pi$ rad to $\pi$ rad and having period $\pi$~rad.}
\end{figure}

\subsection{Check of applicability of the probabilistic model to input data (consistency)}
Following the guidelines in \cref{process-stc-sec:data_analysis_guidelines_check_consistency} in \cite{Prochazka2022statphys_process_stc} one can check that the input data of photon numbers discussed in \cref{sec:data_example} satisfy the basic inequalities given by \cref{eq:Ni_to_N0_interval_chain_simplified_pol}, and the relative photon numbers $N_{i}(\vec{\rotangle})/N_{0}$ do not depend on $N_{0}$. E.g., if the relative number of photons is determined from the relative beam intensity, see \cref{eq:relation_N_to_I_i_to_0}, then \cref{as:energy_pol} can be tested experimentally for a photon beam passing through surface $\surface_i$ ($i=0,...,3$). The \cref{as:energy_pol} is assumed to be satisfied in our case (see \cref{sec:data_example}). The other assumptions (see \cref{sec:model}) and consequences of the probabilistic model must be tested indirectly.

\subsection{\label{sec:parameterizations}Parameterization of unknown functions}
According to \cref{rmk:3_functions_to_be_parameterized} the used probabilistic model contains 3 unknown functions \parameterizedFunctionspol\ which can be parameterized and determined on the basis of experimental data. Parameterizations of some a priory unknown functions in any model are typically accompanied by several additional implicit assumptions. Let us, therefore, formulate for completeness the following assumption: 
\begin{assumption}[Parameterization of unknown functions]
\label{as:parameterizations_pol}
Let functions \parameterizedFunctionspol\ be parameterized.
\end{assumption}

\subsubsection{Parameterization of probability density \texorpdfstring{$\rhoSpol{0}$}{rhoS0}}
If decrease of measured photon numbers behind the first polarizer $N_{1}(\vec{\rotangle})/N_{0}$ does not depend on the orientation of its polarization axis, see \cref{eq:N_i_over_N_0_i1_ideal_polarizer}, then one may assume that the probability density function $\rhoSargipol{0}$ does not depend on the photon polarization. I.e., the initial photon polarization states are distributed uniformly (the initial photon beam is \emph{unpolarized}, see \cref{eq:rhoS_uniform})
\begin{equation}
		\rhoSargipol{0} = \frac{1}{2\pi} \, . \label{eq:parameterization_rhoS0}
\end{equation}
In more general case it would be necessary to introduce a parameterization depending on $\polarization[i=0]$ and some free parameters; it would correspond to \emph{polarized} light.

\subsubsection{\label{sec:parameterization_ProbTpol}Parameterization of probability \texorpdfstring{$\ProbTi{}$}{PT}}
We may assume that the probability of photon transmission through one polarizer depends only on the difference of the polarizer rotation and the photon polarization ($\ProbTidentityargpol = \ProbTi{}(\polarization[i=\text{in}] - \rotangle)$). The following parameterization may be chosen
\begin{equation}
		\ProbTidentityargpol  = 1 - \frac{1 - g(\polarization[i=\text{in}],\rotangle)}{1+a_{2}g(\polarization[i=\text{in}],\rotangle)} \label{eq:parameterization_ProbTpol}
\end{equation}
where 
\begin{equation}
		g(\polarization[i=\text{in}],\rotangle) = \e{-a_0 \left(\frac{\sym(\polarization[i=\text{in}] - \rotangle)}{u_0}\right)^{a_1}}
\end{equation}
and $a_0$, $a_1$ and $a_2$ are free parameters ($u_0 = 1 \text{rad}$). The function $\ProbTi{}$ has meaning of probability; its values should be, therefore, in the interval from 0 to 1 (not necessary in the full interval) for given values of the free parameters.

\subsubsection{\label{sec:parameterization_rhoCpol}Parameterization of probability density \texorpdfstring{$\rhoCi{}$}{rhoC}}
Parameterization of the function $\rhoCargpol$ may be chosen in the form of Gaussian function
\begin{equation}
\rhoCargpol =  \frac{1}{2\sigma\sqrt{\pi}\erf\left(\frac{\pi}{2\sigma}\right)} \e{-\left(\frac{\sym(\polarization[i=\text{out}] - \rotangle)}{\sigma}\right)^2}
\label{eq:parameterization_rhoCpol}
\end{equation}
where $\sigma$ is a free parameter and $\erf(x)$ is the error function defined as
\begin{equation}
		\erf(x) = \frac{2}{\sqrt{\pi}} \int_0^{x} \e{-t^2} \text{d} t  \, .
\end{equation}
The probability density function $\rhoCi{}$ given by \cref{eq:parameterization_rhoCpol} is normalized to 1 when integrated over \polarization[i=\text{out}] (required by the normalization condition given by \cref{eq:rhoargstateC_norm}). It does not depend on the value of \polarization[i=\text{in}] which is consistent with \cref{as:rhoC_independent_of_polarization_in_pol}. It is assumed that it depends only on the difference $\polarization[i=\text{out}] - \rotangle$ (similarly as in the case of the parameterization of the $\ProbTi{}$ function, see \cref{eq:parameterization_ProbTpol}).  

The parameterization of $\rhoCi{}$ given by \cref{eq:parameterization_rhoCpol} corresponds to continuous spectrum of polarization angle values $\polarization[i=\text{out}]$ centered around the value given by the rotation of the polarizer $\rotangle$. If only a single discrete value $\polarization[i=\text{out}]$ (equal to $\rotangle$) were admitted then the probability density function would be represented by corresponding delta function. The smaller the value of the free parameter $\sigma$ (i.e., the width of the corresponding peak), the closer the continuous spectrum is to the delta function.

\begin{remark}
The parameterizations of $\ProbTi{}$ and $\rhoCi{}$ given by \cref{eq:parameterization_ProbTpol,eq:parameterization_rhoCpol} correspond to the definition of linear polarizer, see \cref{def:linear_polarizer_prob}. The parameterizations of the 3 functions $\rhoSpol{0}$, $\ProbTi{}$ and $\rhoCi{}$ a priory restrict set of possible solutions which could describe measured data. The parameterizations represent additional assumptions in the probabilistic model, see \cref{as:parameterizations_pol}. The parameterizations are consistent with the assumptions of stochastic process given by \cref{def:process_polarizers_case_3}.
\end{remark}
\begin{remark}
Four free parameters have been introduced ($\sigma$, $a_0$, $a_1$ and $a_2$). Their physical meaning is not important, if we stay in the first stage of data analysis, see \cref{sec:theory_of_stochastic_processes_adventages}.
\end{remark}

\subsection{Fitting of the probabilistic model to data}
One can try to determine the parameterized functions (values of all the free parameters) in stochastic process given by \cref{def:process_polarizers_case_3} on the basis of experimental data by means of optimization techniques. The relative photon numbers $\frac{N_{i+1}(\vec{\rotangle})}{N_{0}}$ can be calculated with the help of \cref{eq:N_ip_over_N0_case_3}. The quantity $\frac{N_{i+1}(\vec{\rotangle})}{N_0}\rhoSargipol{i+1}$ can be calculated using \cref{eq:dos_ip_over_N0_case_3} or \cref{eq:dos_ip_over_N0_case_3_simplified} which is less computationally intensive task, see also \cref{rmk:simpler_numeric_2} and  \cref{process-stc-sec:computational_complexity} in \cite{Prochazka2022statphys_process_stc} for further comments related to computational complexity. Calculation of the relative photon numbers $N_{3}(\vec{\rotangle})/N_{0}$ behind the third (i.e., the last) linear polarizer in the sequence according to \cref{eq:N_ip_over_N0_case_3}, and needed for comparison to data, is computationally the most intensive task. 

Calculated relative photon numbers $N_{i}(\vec{\rotangle})/N_{0}$ ($i=1,...,M$) behind $(i-1)$-th linear polarizer can be calculated for any value of the rotation angle of the polarizer and the rotations of all the preceding polarizers, i.e., $i$ continuous parameters represented by $\vec{\rotangle}$. The rotation angles are typically measured in discrete steps and the example data, see \cref{sec:data_example}, can be considered only in discrete steps, too. However, even if some discrete angular steps are considered it may still represent far too many data points (depending on the width of the steps) already in the case of $\ntrans=3$. It is, therefore, useful to simplify fitting of the model to data as much as possible.

It is not necessary to take into account all values of $\rotangle[i=0]$ due to the symmetric dependence of the relative number of transmitted particles on the value of $\rotangle[i=0]$ (see \cref{sec:analysis_initial_exploration}). It is sufficient to consider only one value, e.g., $\rotangle[i=0]=0\text{ deg}$. This is closely related to the fact that the initial light is taken as unpolarized, see \cref{eq:parameterization_rhoS0}.

\subsection{\label{sec:results}Numerical results}
\begin{table}
\captionsetup{width=0.95\columnwidth}
\centering
\begin{tabular}{lcc}

\hline\hline
$a_{0}$                         &              & 1.693     \\  
$a_{1}$                         &              & 2.546     \\  
$a_{2}$                         &              & 0.4807    \\  
$\sigma$                        & [rad]        & 0.1887    \\  
\hline \hline
\end{tabular}\\ 
\caption{\label{tab:pars_and_quantities}The values of the free parameters of the probabilistic model determined on the basis of data using optimization techniques.}
\end{table}

It is possible to find dependence of the parameterized functions $\ProbTi{}$ and $\rhoCi{}$ (i.e., values of the free parameters discussed in \cref{sec:parameterization_ProbTpol,sec:parameterization_rhoCpol}) which can describe the input data, under the given set of assumptions on which the stochastic process given by \cref{def:process_polarizers_case_3} is based. The solution is, however, not unique. Therefore, only one solution corresponding to the values of the free parameters in \cref{tab:pars_and_quantities} is discussed in the following. This solution corresponds to the highest value of the parameter $\sigma$ for which one can still fit well the input data. We will come back to the problem of ambiguity in \cref{sec:data_analysis_open_questions}. 

The comparison of the measured number of transmitted photons to calculated (simulated) number of transmitted photons with the help of the probabilistic model is in \cref{fig:model_data_comparison}. The model agrees well with the input data and can be further improved by using more flexible parameterizations of the functions $\ProbTi{}$ and $\rhoCi{}$. Our focus is, however, on conceptually important points and questions, as it has been already mentioned.

The probability $\ProbTi{}$ given by \cref{eq:parameterization_ProbTpol} as a function of $\polarization[i=\text{in}]$ is plotted in \cref{fig:p_T_polarization_simplified_i_0} for $\rotangle[i=0]=0\text{ deg}$. The values of this function representing probability are in the interval from 0 to 1.

%
%
%

\begin{figure}
\captionsetup{width=0.95\columnwidth}
\centering
\includegraphics*[width=\columnwidth]{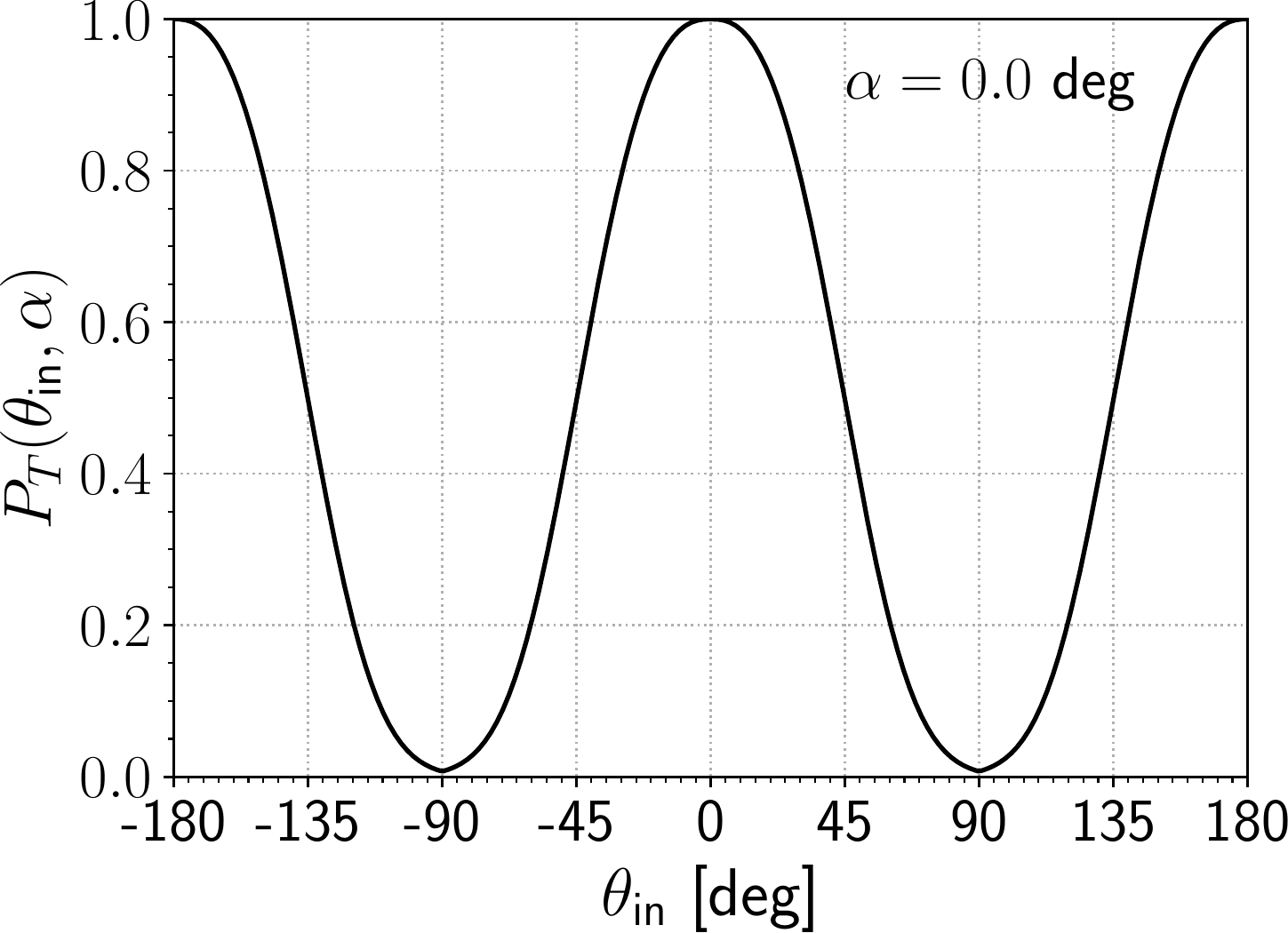}
		\caption{\label{fig:p_T_polarization_simplified_i_0}The probability of transmission of a photon through one polarizer having the orientation of its axis $\rotangle=0\text{ deg}$ as a function of the polarization of the incoming photon.}
\end{figure}

\begin{figure}
\captionsetup{width=0.95\columnwidth}
\centering
\includegraphics*[width=\columnwidth]{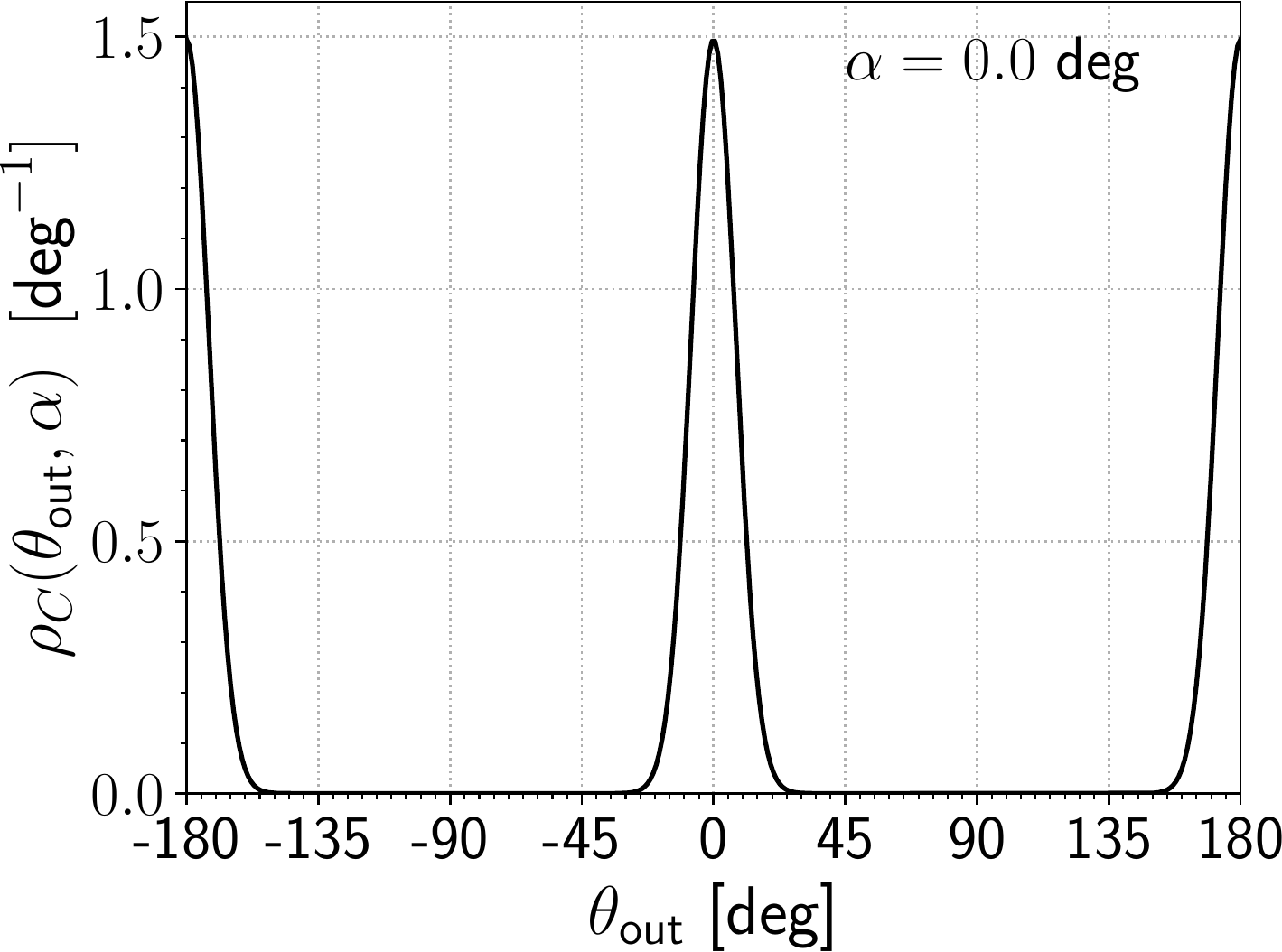}
\caption{\label{fig:rho_C_fix_i_0_rotangle_0}The probability density function of change of photon polarization of a photon transmitted through one polarizer having the orientation of its axis $\rotangle=0\text{ deg}$ as a function of the polarization of the outgoing photon.}
\vspace{5mm}
\end{figure}

\newcommand{\commoncaption}[1]{The comparison of the input data of relative photon numbers $N_{i}(\vec{\rotangle})/N_0$ ($i={#1}$) corresponding to transmission of unpolarized light through 3 ideal identical linear polarizers (Malus's law) with the probabilistic model (see \cref{sec:model_case_3}) fitted to the data for several combinations of rotations of the polarizes $\vec{\rotangle}$ (in degrees). Blue lines - the input data given by \cref{eq:N_i_over_N_0_ideal_polarizer}. Orange lines - the probabilistic model, see \cref{eq:N_ip_over_N0_case_3,eq:dos_ip_over_N0_case_3}.}
\begin{figure*}
\centering
\captionsetup{width=\textwidth}
\begin{subfigure}[t]{.48\textwidth}
\centering
\includegraphics*[width=\textwidth]{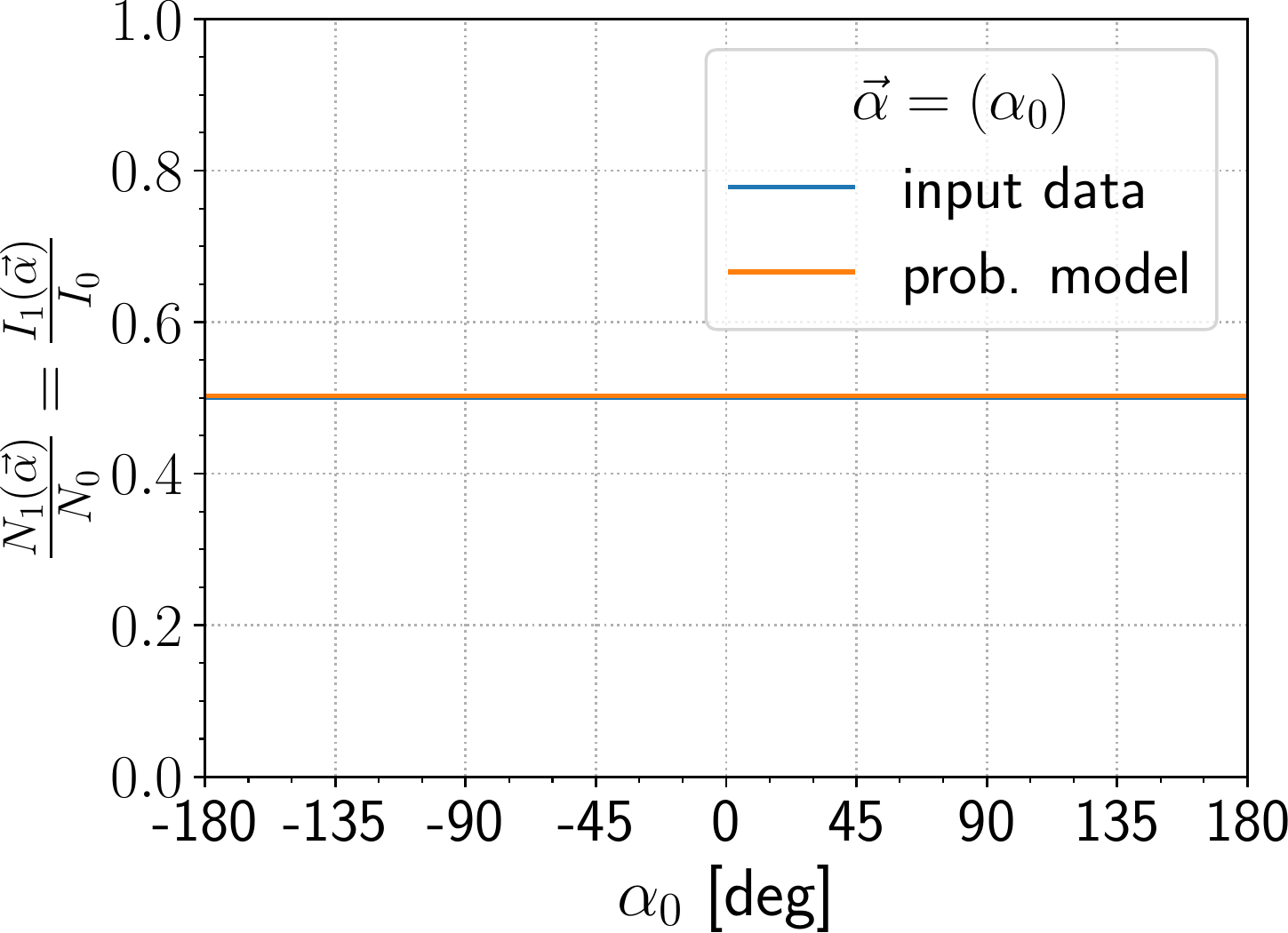}
\caption{\label{fig:N1_i_1_j_0_theta_0} $N_{1}(\vec{\rotangle})/N_0$ as a function of $\rotangle[i=0]$.}
\end{subfigure}
\quad
\begin{subfigure}[t]{.48\textwidth}
\centering
\includegraphics*[width=\textwidth]{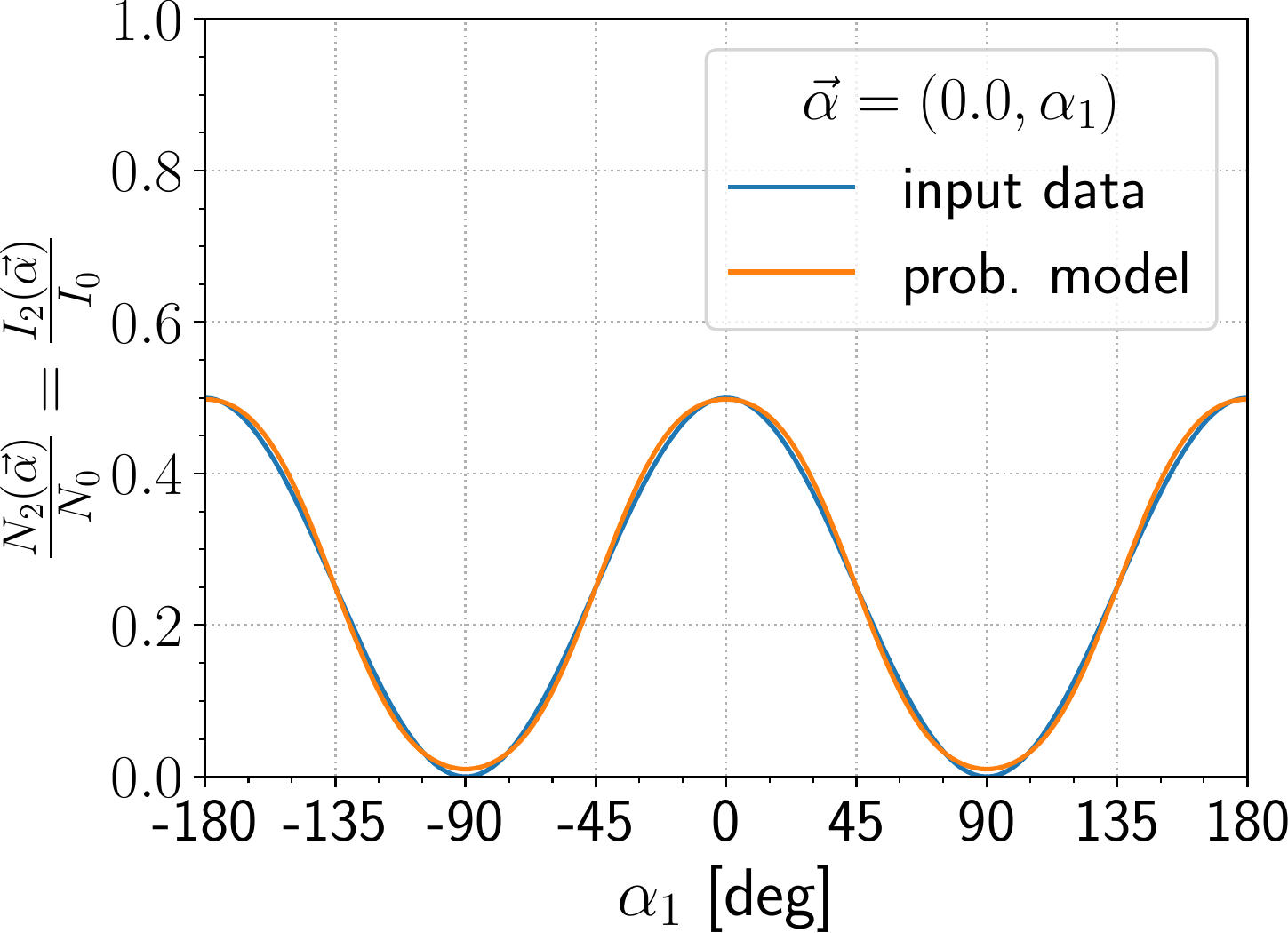}
\caption{\label{fig:N2_fix_i_2_j_1_theta_0_0} $N_{2}(\vec{\rotangle})/N_0$ as a function of $\rotangle[i=1]$ for fixed value of $\rotangle[i=0]=0\text{ deg}$.}
\end{subfigure}
\centering
\begin{subfigure}[t]{.48\textwidth}
\centering
		\includegraphics*[width=\textwidth]{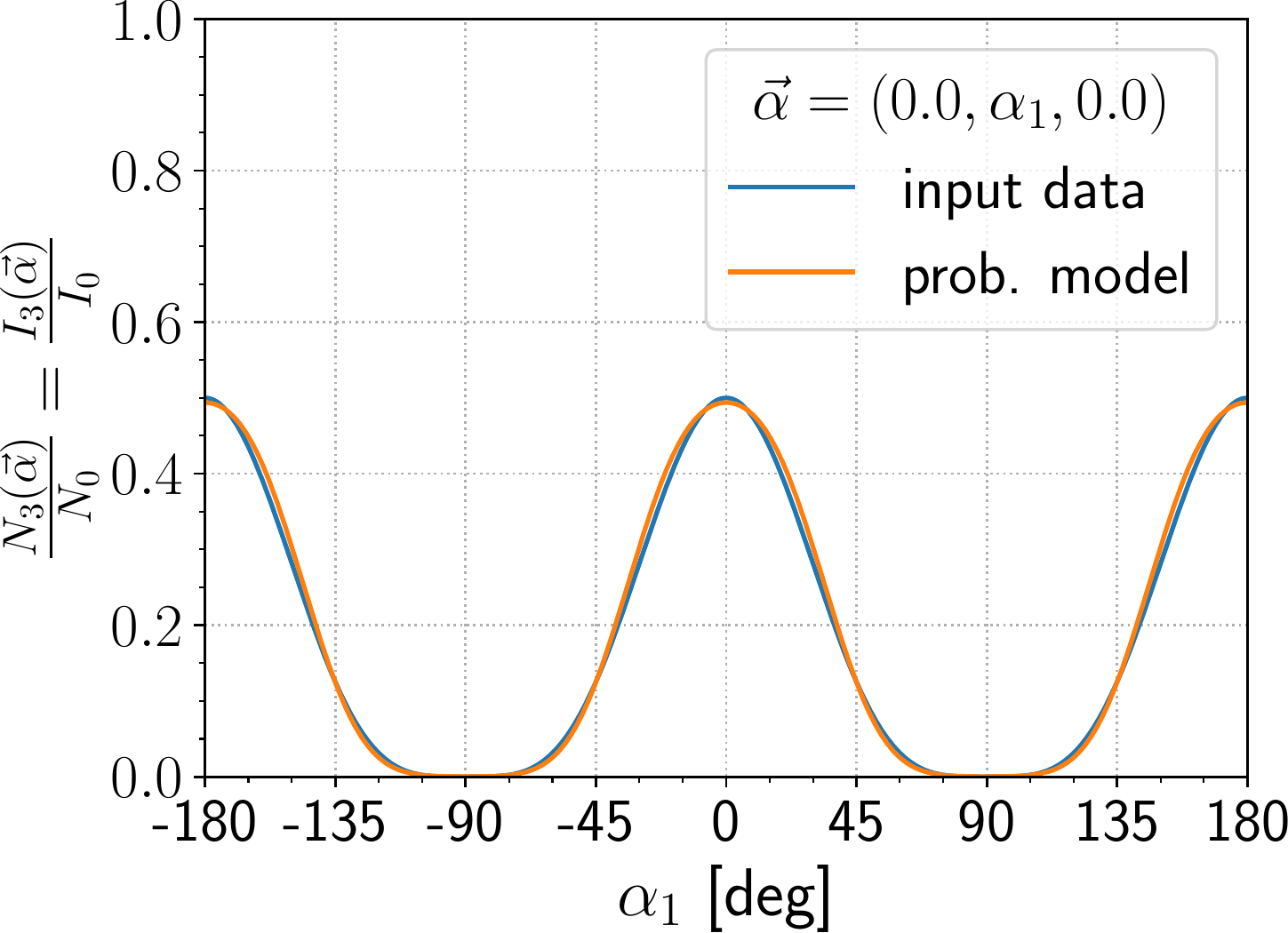}
\caption{\label{fig:N3_fix_i_3_j_1_theta_0_0_0} $N_{3}(\vec{\rotangle})/N_0$ as a function of $\rotangle[i=1]$ for fixed values of $\rotangle[i=0]=0\text{ deg}$ and $\rotangle[i=2]=0 \text{ deg}$.}
\end{subfigure}
\quad
\begin{subfigure}[t]{.48\textwidth}
\centering
\includegraphics*[width=\textwidth]{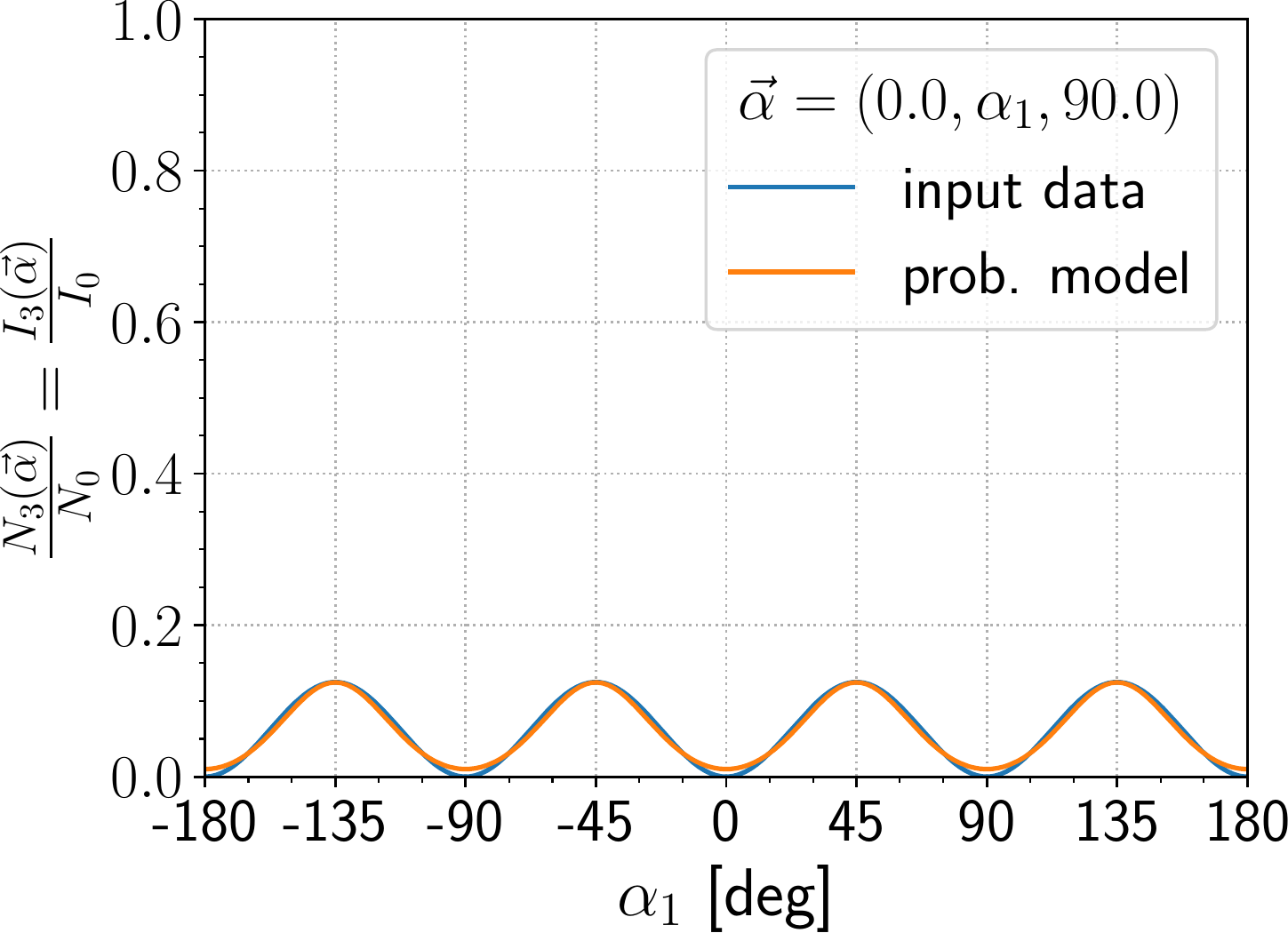}
\caption{\label{fig:N3_fix_i_3_j_1_theta_0_0_1.5707963267948966} $N_{3}(\vec{\rotangle})/N_0$ as a function of $\rotangle[i=1]$ for fixed values of $\rotangle[i=0]=0\text{ deg}$ and $\rotangle[i=2]=90 \text{ deg}$.}
\end{subfigure}
\caption{\label{fig:model_data_comparison} \commoncaption{1,2} 
}
\end{figure*}

\begin{figure*}
\begin{minipage}[t]{.96\textwidth}
\centering
\captionsetup{width=\columnwidth}
\includegraphics*[width=\textwidth]{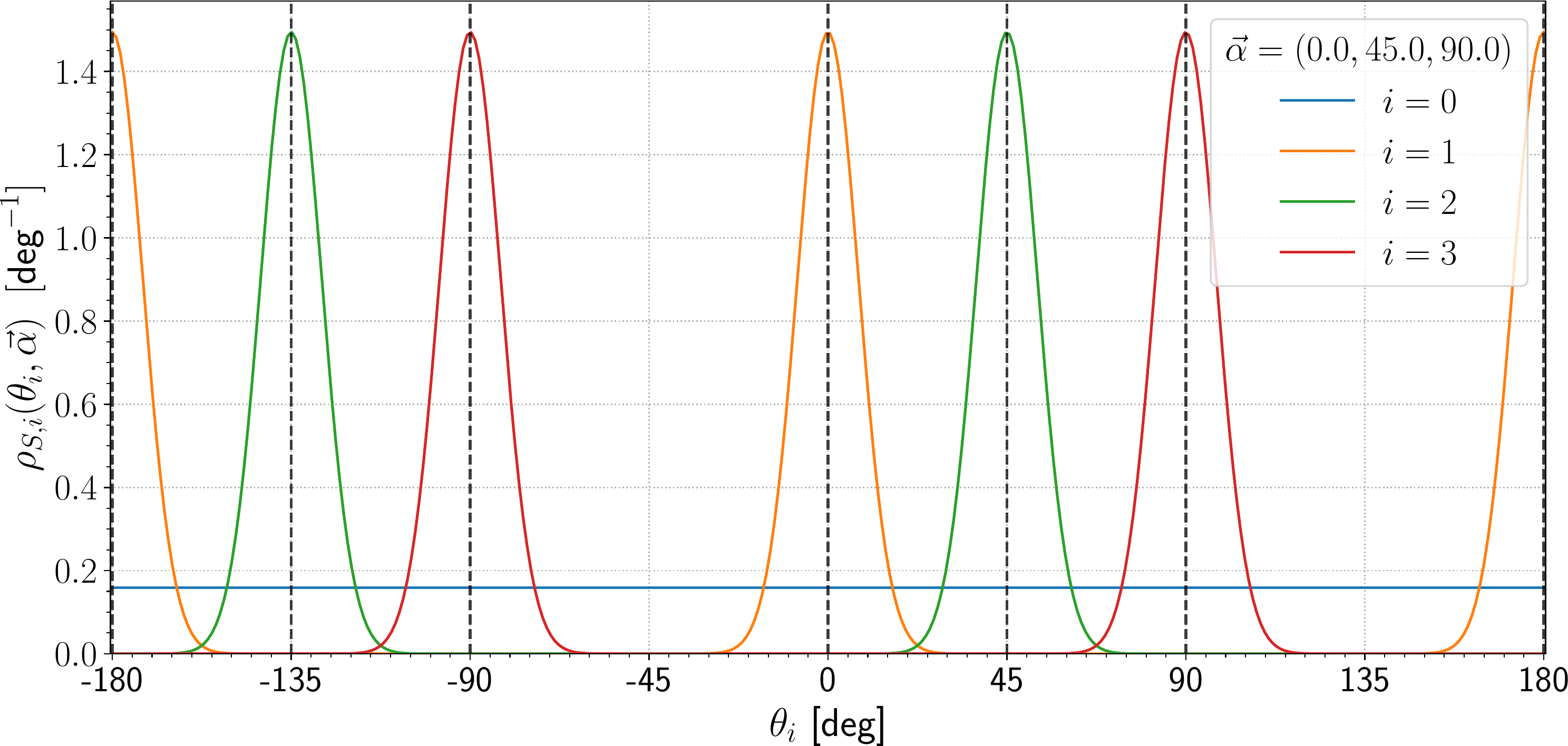}
\caption{\label{fig:rho_S_polarization_group_theta}The probability density functions $\rhoSargipol{i}$ as a function of photon polarization angle $\polarization[i=i]$ if the 3 ideal identical linear polarizers are rotated by $\vec{\rotangle} = (0.0, 45.0 \text{ deg}, 90.0 \text{ deg})$. In the case of $i=0$ the function represents initial uniform polarization of photons in the beam given by \cref{eq:parameterization_rhoS0}, in the case of $i=1,2,3$ it describes polarization of the beam transmitted through the one, two or three polarizers. The orientations of the axes of the polarizers are specified by $\vec{\rotangle}$, see the vertical dashed lines. The position of the peaks correspond to the rotation of the axis of the polarizers.}
\end{minipage}
\end{figure*}
\begin{figure*}
\begin{minipage}[t]{.96\textwidth}
\centering
\captionsetup{width=\columnwidth}
\includegraphics*[width=\textwidth]{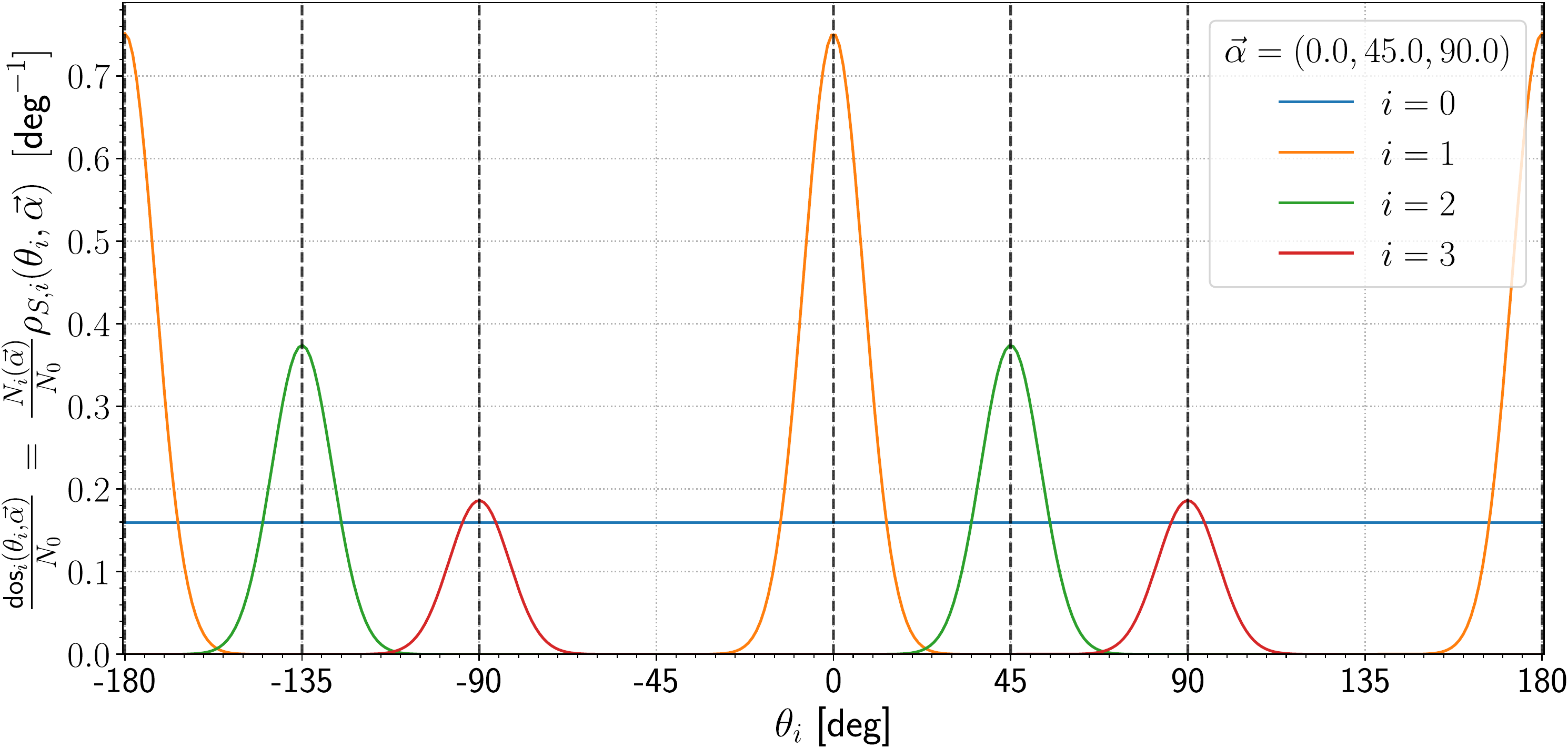}
		\caption{\label{fig:dos_polarization_group_theta}Similar picture to \cref{fig:rho_S_polarization_group_theta} but the probability density functions $\rhoSargipol{i}$ multiplied by $N_{i}(\vec{\rotangle})/N_0$ are plotted, these expressions are equal to relative densities of states $\dosargipol{i} / N_0$, see \cref{eq:def_rel_flux_EventSXw_i_rho_simplified}.}
\end{minipage}
\end{figure*}


The probability density function $\rhoCi{}$ given by \cref{eq:parameterization_rhoCpol} as a function of $\polarization[i=\text{out}]$ is plotted in \cref{fig:rho_C_fix_i_0_rotangle_0} for $\rotangle[i=0]=0\text{ deg}$. This function is normalized to unity, see \cref{eq:rhoargstateC_norm}, and the values are not in the interval from 0 to 1 (values of a probability density functions may or may not be in the interval from 0 to 1). 


The probability density functions $\rhoSargipol{i}$ as a function of photon polarization angle $\polarization[i=i]$ ($i=0,...,3$) are plotted in \cref{fig:rho_S_polarization_group_theta} for one fixed combination of rotations of the polarizes: $\vec{\rotangle} = (0.0, 45.0 \text{ deg}, 90.0 \text{ deg})$. The blue line in \cref{fig:rho_S_polarization_group_theta} corresponds to \cref{eq:parameterization_rhoS0}, i.e., to the assumed initial uniform distribution of photon polarization (unpolarized light). Unpolarized light transmitted through one or more linear polarizers is not uniform, see \cref{fig:rho_S_polarization_group_theta} (and \cref{eq:parameterization_rhoCpol}). The probability density functions $\rhoSargipol{i}$ ($i=1,...,M$) have the same shapes (differing only in position). It is the consequence of the used assumptions, see \cref{rmk:same_shapes}.

The probability density function $\rhoSargipol{i}$ multiplied by $N_{i}(\vec{\rotangle})/N_0$ ($i=0,..,3$) for the same rotations of the polarizes $\vec{\rotangle} = (0.0, 45.0 \text{ deg}, 90.0 \text{ deg})$ is plotted in \cref{fig:dos_polarization_group_theta} similarly as the functions $\rhoSargipol{i}$ are plotted in \cref{fig:rho_S_polarization_group_theta}. The function $(N_{i}(\vec{\rotangle})/N_0) \; \rhoSargipol{i}$ is equal to $\dosargipol{i}/N_0$ (see \cref{eq:def_rel_flux_EventSXw_i_rho_simplified}). This function expresses not only the density of polarization angles but also the decrease of the photon numbers behind $(i-1)$-th polarizer (if $i=1,2,3$), see \cref{eq:def_rel_flux_EventSXw_i_rho_simplified,eq:Ni_over_N0_given_by_integration_simplifed}. For the value of $\vec{\rotangle}$ it holds $N_1(\vec{\rotangle})/N_0=0.5$, $N_2(\vec{\rotangle})/N_0=0.25$, $N_3(\vec{\rotangle})/N_0=0.12$ (according to both the input data and the probabilistic model fitted to the data). This decrease of number of transmitted photons with increasing number of polarizers is visible in \cref{fig:dos_polarization_group_theta}, too.


The numerical results clearly show that it is possible to explain the famous 3 polarizers experiment using the idea of quanta of light (photons) and the theory of stochastic processes. It is possible to determine, with the help of the stochastic process given by \cref{def:process_polarizers_case_3}, the probability (density) functions characterizing the transmission of light through the polarizers mentioned in \cref{sec:introduction}. The open questions related to the ambiguity of the determination of the parameterized functions will be discussed in \cref{sec:data_analysis_open_questions}.

The numerical results presented in this section have been obtained with the help of ROOT \cite{Brun1997} and Matplotlib \cite{Hunter2007}.

\clearpage
\FloatBarrier
\subsection{\label{sec:data_analysis_open_questions}Open questions}
As it has been mentioned, the determined dependence of the parameterized functions $\ProbTi{}$ and $\rhoCi{}$ shown in \cref{sec:results} can explain the measured (input) number of transmitted photons through 3 ideal identical linear polarizers, but it is not unique. There are several open questions concerning determination of the functions:
\begin{openquestions}
\item{\label{oq:rhoCpol_outgoing_polarization}\textit{Dependence of $\rhoCi{}$ on $\polarization[i=\text{out}]$}\\
There is strong ambiguity on determination of the value of the parameter $\sigma$, i.e., the width of the peaks characterizing the function $\rhoCi{}$, see \cref{fig:rho_C_fix_i_0_rotangle_0}. The input data could be fitted similarly well with value of $\sigma$ being in the interval from zero to the value specified in \cref{tab:pars_and_quantities}. 
}
\item{\label{oq:rhoCpol_symmetry_one}\textit{$\rotangle$ and ($\rotangle - \pi$) (a)symmetry of $\rhoCi{}$ as a function of $\polarization[i=\text{out}]$}\\
The parameterization of probability density function $\rhoCi{}$ given by \cref{eq:parameterization_rhoCpol} (see also \cref{fig:rho_C_fix_i_0_rotangle_0}) implies that $\rhoCi{}$ as a function of polarization angle of outgoing photon $\polarization[i=\text{out}]$ has two peaks located at $\rotangle$ or $\rotangle - \pi$. Moreover, the peaks have the same shape. There is, therefore, the same probability that an outgoing photon has value of polarization angle $\polarization[i=\text{out}]$ closer to either $\rotangle$ or $\rotangle - \pi$.
This symmetry is clearly visible in \cref{fig:dos_polarization_group_theta}, too. 

The input data analyzed in \cref{sec:data_analysis} can be, however, fitted equally well with similar parameterizations of $\rhoCi{}$ normalized to unity and having only one peak or two peaks of different heights and widths. The given input data do not allow distinguishing between these possibilities. 
}
\item{\label{oq:rhoCpol_symmetry_two}\textit{$\rotangle + \pi/2$ and $\rotangle - \pi/2$ (a)symmetry of $\rhoCi{}$ as a function of $\polarization[i=\text{out}]$}\\
		There is one more source of ambiguity. The one or two peaks in the spectrum of values of polarization angle $\polarization[i=\text{out}]$ of outgoing photon may be located at different positions, at $\rotangle + \pi/2$ and $\rotangle + \pi/2 - \pi = \rotangle - \pi/2$. This situation corresponds to the axis of polarizer specified by angle $\rotangle + \pi/2$ which is perpendicular to the axis specified by angle $\rotangle$. The number of transmitted photons through 3 ideal identical polarizers (see \cref{sec:data_example}) exclude the possibility that the spectrum could have two, three or four peaks corresponding to \emph{both} the axes. The peaks must correspond to only one axes. It is, however, not possible to determine on the basis of the input data towards which of the axes the polarization of an outgoing photon is predominantly aligned. 
}
\item{\label{oq:rhoCpol_dependence}\textit{Dependence of $\rhoCi{}$ on $\polarization[i=\text{in}]$}\\ 
It has been assumed that the parameterization of $\rhoCi{}$ given by \cref{eq:parameterization_rhoCpol} does not depend on $\polarization[i=\text{in}]$ and depends only on the difference $\polarization[i=\text{out}] - \rotangle$. The possibility of more complex dependence of the function $\rhoCi{}$ characterizing given polarizer on $\polarization[i=\text{in}]$, $\polarization[i=\text{out}]$ and $\rotangle$ should be tested.
}
\item{\label{oq:ProbCT_dependence}\textit{Probability function $\ProbTi{}$}\\
The parameterization of $\ProbTi{}$ given by \cref{eq:parameterization_ProbTpol} depends only on the difference $\polarization[i=\text{in}] - \rotangle$. The possibility of more complex dependence of $\ProbTi{}$ on $\polarization[i=\text{in}]$ and $\rotangle$ should be tested, too.
}
\end{openquestions}
\begin{remark} 
Some structures (e.g., needle-like crystals or molecules) in a linear polarizer can be aligned predominantly in parallel to one axis of the polarizer. This corresponds to one axis of symmetry. The second axis of symmetry is perpendicular to it. One may ask how the quality of alignment of these microscopical structures is related to the width of the peaks in the spectrum of values of polarization angles $\polarization[i=\text{out}]$ of outgoing photons characterized by the parameter $\sigma$. This is, however, an example of a question which goes beyond the scope of the first stage of analysis of data.

\end{remark}

\section{\label{sec:open_questions_answers}Possibilities of unique determination of parameterized probability (density) functions}
If a probabilistic model does not allow to describe measured data under given set of assumptions then it means that there is a contradiction which must be removed (it may be, e.g., necessary to modifying one or more assumptions of the model). The analysis of data performed in \cref{sec:data_analysis} is an example of analysis where input data partially constrain dependences of the parameterized functions, but unique determination of the free parameters (the parameterized functions) is not obtained. This may be in many cases very useful and valuable information. However, there are situations when one would like to determine uniquely the parameterized functions, or at least better constrain them on the basis of experimental data. 

Possibilities of obtaining more experimental data are discussed in \cref{sec:open_questions_answers_more_data}. In some cases these data can be analyzed with the help of stochastic IM process, this is covered in \cref{sec:open_questions_answers_IM_process}. The other cases are discussed in \cref{sec:open_questions_answers_beyond_IM_process}.

\subsection{\label{sec:open_questions_answers_more_data}Suitable experimental data}
As to more experimental data needed to determine uniquely the parameterized function in \cref{sec:data_analysis}, measuring transmission of light through various sequences of polarization sensitive elements (see \cref{process-stc-sec:problem_statement} in \cite{Prochazka2022statphys_process_stc}) can provide essential experiment data which can be further analyzed with the help of the probabilistic approach. This should help to uniquely determine the functions $\ProbTi{}$ and $\rhoCi{}$ characterizing individual polarization sensitive elements. Similar measurements are needed for determination of polarization states of photons emitted by a source, i.e., for unique determination of the function $\rhoSi{0}$.

E.g., the \cref{oq:rhoCpol_symmetry_one,oq:rhoCpol_symmetry_two} could be answered with the help of optical analog of Stern-Gerlach experiment separating photons by their spins (polarization states). One possible experimental method of A.~Fresnel is discussed in detail in \cite{Arteaga2019} (it uses quartz polyprism). 

In the case of optical phenomena not related to polarization of light one can use very similar approach to obtain required experimental information.

\subsection{\label{sec:open_questions_answers_IM_process}Probabilistic models based on IM process}
In some cases an analysis of measured number of transmitted photons through sequences of various polarization sensitive elements can be performed with the help of one of the stochastic processes introduced in \cref{sec:model} with very little or no modification. E.g., transmission of photons through a Faraday-effect-based device can be analyzed very similarly as transmission of photons through a linear polarizer. Instead of rotation angle of a linear polarizer one can introduce, e.g., magnetic field $B$ in the direction of propagation of the beam and the length $d$ of the path where the beam and magnetic field interact in the device (to explain measured number of transmitted photons in dependence on these two parameters, or one of them may be taken as fixed). The functions $\ProbTi{i}$ and $\rhoCi{i}$ for given index $i$ corresponding to a Faraday-effect-based device must be parameterized differently than in the case corresponding to a linear polarizer, see \cref{eq:parameterization_ProbTpol,eq:parameterization_rhoCpol}. I.e., mainly the \cref{as:Xw,as:rhoC_independent_of_polarization_in_pol,as:parameterizations_pol} (and also the \cref{as:rhoC_identity_pol,as:ProbT_identity_pol}, if the beam is not transmitted through sequence of identical Faraday rotators) must be modified. In the case of a Faraday-effect-based device one can try to parameterize the corresponding functions $\ProbTi{}$ and $\rhoCi{}$ according to \cref{def:faraday_rotator_prob} and determine them on the basis of experimental data. The determined position of the peak of $\rhoCi{}$, mentioned in \cref{def:faraday_rotator_prob}, as a function of the difference $\polarization[i=\text{in}]$ - $\polarization[i=\text{out}]$ divided by $B d$ may be then compared to the Verdet constant standardly determined (using a different method without determining the probability (density) functions) to characterize a Faraday-effect-based device.

To describe probabilistic character of transmission of light through a linear polarizer only one variable characterizing properties of photon, the photon polarization angle, has been taken into account in \cref{sec:data_analysis}. Probabilistic description of transmission of light through other polarization sensitive elements where photon polarization states denoted as circular and elliptical are involved may require introduction of one or two additional random variables. In this case or in the case of adding some other random or non-random variables (parameters) characterizing photon states and optical elements the formulae in \cref{sec:model} can be modified in very straightforward way and then used for analysis of experimental data. The formulae remains essentially the same, only photon polarization angle is replaced by other random variables characterizing states of given system, and rotations of elements are replaced by parameters characterizing the states of the system.


Stochastic IM process can be helpful in description of other optical phenomena not related to polarization at all or related to polarization and some other properties of light, if the two main assumption of IM process (concerning memoryless (Markov) property and independence of outcomes of an experiment) are satisfied. Several types of optical experiments are mentioned in \cref{process-stc-sec:problem_statement} in \cite{Prochazka2022statphys_process_stc} where it is discussed that these experiments often include phenomena involving splitting of a photon beam or merging of photon beams (very common in optical systems). In this case there is in general a net of transitions instead of a sequence of transitions, see sect.~3.3 in \cite{Prochazka2022statphys_process_stc}. Coincidence and anti-coincidence experiments common in optics can be also described using IM process. Description of reflection of light on a surface is essentially the same (from the mathematical point of view) as description of transmission of light through an optical element. Experiments with "interference" of light, such as well known double-slit experiment, can be described using IM process, if the observed pattern does not depend on the intensity of the source of light (i.e., the interactions of individual photons with the slits, or another optical element, are independent). 

It may happen that measured numbers of transmitted photons (states of a system) depend not only on some non-random variables, but also on some random variables. I.e., densities of states are essentially measured. In this case it is necessary to modify some of the formulae in \cref{sec:model} (with the help of more general formulae in \cite{Prochazka2022statphys_process_stc}. It concerns mainly the key system of (integral) equations, see \cref{rmk:system_of_equations}.

\subsection{\label{sec:open_questions_answers_beyond_IM_process}Probabilistic models not based on IM process}
Not all optical phenomena may be described with the help of IM process based on the two main assumptions (see \cref{rmk:main_assumptions}). Dynamical effects such as saturation of optical media (typical, e.g., for experiments with stimulated emission) cannot be described under these assumptions. If the intensity of light transmitted through an optical device is too high, then it can completely change optical properties of the device during the transmission (transmission of individual photons may not be independent). In such cases the theory of stochastic can be also very useful, but it is necessary to formulate suitable stochastic process differing from IM process.

\section{\label{sec:conclusion}Summary and conclusion}
Transmission of a photon through an optical element is standardly denoted as probabilistic process. The probability of transition of an input photon state (before an interaction with an element) to an output photon state (after the interaction with the element) can be factorized into 3 probabilistic effects represented by the probability (density) functions $\rhoSi{i}$, $\ProbTi{i}$ and $\rhoCi{i}$ with the help of IM process, see \cref{sec:problem_statement}. Statistical theoretical descriptions widely used in optics and in physics in general do not allow to easily determine these 3 probabilistic effects, \cref{sec:contemprorary_theories}.

The theory of stochastic processes provides general abstract framework for description of random processes. IM process proposed in \cite{Prochazka2022statphys_process_stc} and formulated within the framework of the theory of stochastic processes is suitable for description of various processes in physics, see \cref{sec:theory_of_stochastic_processes}. 

It has been shown that formulae suitable for description of transmission of photons through polarization sensitive elements including, but not limited to, 3 linear polarizers can be derived with the help of IM process, see \cref{sec:model}. Probability density function $\rhoSi{}$ allows to define polarized and unpolarized light. The probability $\ProbTi{}$ and probability density function $\rhoCi{}$ can be used for definition of various types of optical elements. These new definitions based on probability (density) functions are discussed in \cref{sec:definitions}. In similar way photon beam and other optical elements not related to polarization of light may be characterized, too.

Stochastic IM processes suitable for description of polarization of light formulated in \cref{sec:model} can be used for an analysis of measured relative photon numbers passed through $\ntrans$ polarization sensitive elements in dependence on the rotations of the elements. Basic aspects of the measurement are in \cref{sec:measurement}. It is worth to note that in these cases some properties of the photon beam are measured (the relative number of transmitted photons) and some other statistical characteristics of the individual photons states and their interactions with matter can be determined with the help of the stochastic process. Example data of three-polarizers experiment are analyzed in \cref{sec:data_analysis}. This is novel type of analysis showing how to determine the probability (density) functions $\rhoSi{i}$, $\ProbTi{i}$ and $\rhoCi{i}$ characterising transmission of individual photon through various optical elements on the basis of experimental data. 

Experiments concerning spin-dependent and polarization-dependent phenomena are essential for better understanding connection between spin and polarization of particles. Better understanding of polarization of light at the level of individual photons can stimulate further (already very broad) usage of various polarization sensitive devices in optical systems. However, the open questions formulated in \cref{sec:data_analysis_open_questions} should be answered before drawing far-reaching conclusions concerning polarization. Suggestions how to address the questions have been discussed in \cref{sec:open_questions_answers}.


The well known 3 polarizers experiment is an example of an experiment that, although not quantitatively described in the literature using the corpuscular theory, can be described with the help of IM process and the corpuscular idea of light. Using IM process (or another stochastic process based on different set of assumptions) it is possible to analyse many experiments and further explore possibilities of the corpuscular theory. With the help of IM stochastic process it is possible to unify description of various phenomena which may look very differently at first glance, see sect.~5 in \cite{Prochazka2022statphys_process_stc}. E.g., IM stochastic process has been successfully used in genetics (biophysics, see \cref{rmk:genetics}), can be used for description of polarization of light (particle-matter interaction) as well as for description of motion of particles having initial conditions characterized by probability density functions, see \cite{Prochazka2022statphys_motion_free}. 

Description of experimental data based on stochastic process like IM process has several advantages over other statistical theoretical approaches used in physics, see \cref{sec:theory_of_stochastic_processes_adventages}. One of the advantages of analysis of experimental data using IM process is that it can be divided into two stages. In the first stage one can try to determine overall probability (density) functions characterizing transitions of states of a system on the basis of measured quantities without the aim to explain in more detail the probabilities. It is not necessary to know, e.g., the Hamiltonian of the system, introduce complex amplitudes, know all microscopic processes etc. In the second stage, one may try find more detailed internal mechanism (causes) of the transitions of the system leading to the determined probabilities. The analysis of data performed in \cref{sec:data_analysis} corresponds to the first stage. 

In some fields of research the theory of stochastic processes and the separation of an analysis of data into the two stages has been essential to make further progress. The field of optics (studying propagation of light and interaction of light with matter) is one of the field where many very diverse theoretical approaches for description of optical phenomena have been studied for centuries. However, the potential of applicability of the theory of stochastic processes has not yet been fully explored in this field. In many analyses of experimental data not event the first stage has been achieved. This is the case of, e.g., analysis of data concerning polarization of light. It has been shown in the presented paper how the theory of stochastic processes and the corpuscular idea of light can open up new possibilities to make further progress in understanding various phenomena in the field of optics. One can leverage terminology, techniques and results already known in the theory of stochastic processes and successfully used in many areas of research.

%
%
%
%
%

\end{document}